\documentclass{llncs}
\usepackage{url,amssymb,amsmath,latexsym,mathrsfs}
\usepackage{tikz}
\usetikzlibrary{decorations.pathmorphing,positioning,shapes,arrows}
\usepackage{todonotes}
\usepackage{thmtools}
\usepackage{hyperref}
\usepackage{thm-restate}
\usepackage{algorithm}
\usepackage{algpseudocode}
\algrenewcommand\alglinenumber[1]{\tiny #1:}

\newcommand{\reals}{\mathbb{R}}

\newcommand{\traj}{\tau}
\newcommand{\trj}{\gamma}

\newcommand{\mc}{M}

\newcommand{\flow}{\Phi}

\newcommand{\move}[1]{\stackrel{#1}{\rightarrow}}
\renewcommand{\path}{\rho}
\newcommand{\node}{\rho}

\newcommand{\expect}{\mathbb{E}}
\newcommand{\prob}{\textup{Pr}}

\def\bx{{\textbf{x}}}

\def\bv{{\textbf{v}}}

\newcommand{\ints}{\mathcal{I}}
\newcommand{\hyper}{\mathcal{H}}
\newcommand{\hcube}{\mathfrak{h}}
\newcommand{\states}{\Upsilon}
\newcommand{\edges}{\Rightarrow}
\newcommand{\pathmc}{\pi}
\newcommand{\state}{\eta}
\newcommand{\TRJ}{TRJ}

\newcommand{\F}{\mathbf{F}}
\newcommand{\U}{\mathbf{U}}
\newcommand{\G}{\mathbf{G}}

\newcommand{\citep}{\cite}

\renewcommand{\Re}{\mathbb{R}}

\newcommand{\vv}{\textbf{v}}
\newcommand{\vx}{\textbf{x}}

\newcommand{\ra}{\rightarrow}
\newcommand{\tm}{\mathbb{T}}
\newcommand{\Ptm}{\textbf{P}_{\mathbb{T}_j(\vv)}}

\newcommand{\tildex}{\raise.17ex\hbox{$\scriptstyle\mathtt{\sim}$}}

\def\bx{{\textbf{x}}}

\def\bv{{\textbf{v}}}

\def\ra{\rightarrow}

\newcounter{theass} 

\newcounter{theeg} 

\newcommand{\INIT}{\textup{INIT}}

\newcommand{\paths}{\mathsf{paths}}

\begin{document}

\title{Approximate probabilistic verification of \\ hybrid systems}


\author{Benjamin M. Gyori\inst{2}
\and Bing Liu\inst{3}
\and Soumya Paul\inst{4}
\and R. Ramanathan\inst{1}
\and \\
P.S. Thiagarajan\inst{1}}
\institute{Department of Computer Science, National University of Singapore, Singapore \and
Department of Systems Biology, Harvard Medical School, USA \and
Department of Computational and Systems Biology, University of Pittsburgh, USA \and
Institute de Recherche en Informatique de Toulouse, 31062 France}

\maketitle

\vspace{-20pt}

\begin{abstract}
Hybrid systems whose mode dynamics are governed by non-linear ordinary differential equations (ODEs) are often a natural model for biological processes. However such models are difficult to analyze. To address this, we develop a probabilistic analysis method by approximating the  mode transitions as stochastic events. We assume that  the probability of making a mode transition is proportional to the measure of the set of pairs of time points and value states  at which the mode transition is enabled. To ensure a sound mathematical basis, we impose  a natural continuity property on the non-linear ODEs. We also assume that the states of the system are observed at discrete time points but that the mode transitions may take place at any time between two successive discrete time points. This leads to a discrete time Markov chain as a probabilistic approximation of the hybrid system. We then show that for BLTL (bounded linear time temporal logic) specifications the hybrid system meets a specification iff its Markov chain approximation meets the same specification with probability $1$. Based on this, we formulate  a sequential hypothesis testing procedure  for verifying -approximately- that the Markov chain  meets a BLTL specification with high probability. Our case studies on cardiac cell dynamics and the circadian rhythm indicate that our  scheme can be applied in a number of realistic settings.
\end{abstract}

\vspace{-20pt}

\\

\vspace{-16pt}

\section{Introduction}\label{sec:intro}
Hybrid systems are often used to model biological processes \cite{batt2005qualitative,bruce2014modelling,buckwar2011exact}. The analysis of these models is difficult due to the high expressive power of the mixed dynamics \cite{henzinger-lics-survey}. Various lines of work have explored  ways to mitigate this problem with  a common technique being to restrict the mode dynamics \cite{girard2006efficient,frehse2005phaver,clarke2003verification,alur2000discrete,agrawal2006behavioural,henzinger1999discrete}.
However, for many of the models arising in systems biology the mode dynamics will be governed by a system of (non-linear) ordinary differential equations (ODEs).  To analyze such systems, we develop a scheme  under which such systems can be approximated as a discrete time Markov chain.

A key difficulty in analyzing a hybrid system's behavior is that the time points and value states at which a trajectory meets a guard will depend on the solutions to the ODE systems associated with the modes. For high-dimensional systems these solutions will not be available in closed form. To get around this,  we assume that the mode transitions are stochastic events and that the probability of a mode transition is proportional to the measure of the value state and time point pairs  at which this  transition is enabled. More sophisticated hypotheses could be considered. For instance one could tie the mode transition probability to how long the guard has been continuously enabled or how deeply within a guard region the current state is. To bring out the main ideas we will postpone exploring such approximations to our future work.

To secure a sound mathematical basis for our  approximation, we further assume: (i) The vector fields associated with the ODEs are $C^1$ (continuously differentiable) continuous functions.(ii) The states of the hybrid system are observable only at discrete time points. (iii) The set of initial states and the guard sets  are bounded open sets.(iv) The hybrid dynamics is strictly non-Zeno in the sense there is uniform upper bound on the number of transitions that can take place in a unit time interval. For technical convenience we in fact assume that  time discretization is so chosen that at most one mode transition takes place between two successive discrete time points.

Under these assumptions, we show that the dynamics of the hybrid system $H$ can be approximated as an infinite state Markov chain $M$. To relate the behaviors of $M$ and $H$, we use BLTL (bounded linear-time temporal logic~\cite{clarke1999model}) to specify time bounded dynamic properties of $H$. We then show that $H$ meets the specification $\psi$--i.e. every trajectory of $H$ is a model of $\psi$--iff $M$ meets the specification  $\psi$ with probability $1$. This allows us to approximately verify interesting properties of the hybrid system using its Markov chain approximation. However, even a bounded portion of $M$ can not be constructed effectively. This is because the transition probabilities of the Markov chain will depend on the solutions to the ODEs associated with the modes, which will not be available in a closed form. In addition, the structure of $M$ itself will be unknown since the states of the chain will be those that can reached with non-zero probability from the initial mode and we can not determine which transitions have non-zero probabilities. To cope with this, we design a statistical model checking procedure to approximately verify that the chain (and hence the hybrid system) almost certainly meets the specification. One just needs to ensure that the dynamics of the Markov chain is being sampled according to underlying probabilities. We achieve this  by randomly generating trajectories of $H$ through numerical simulations in a way that corresponds to randomly sampling the paths of the Markov chain according to its underlying structure and transition probabilities.

In establishing these results, we assume that the atomic propositions in the specification are interpreted over the modes of the hybrid system. Consequently  one can specify patterns of  mode visitations while quantitative properties can be inferred only indirectly and in a limited fashion. Our results however can be extended to handle quantitative atomic propositions (``the current concentration of protein X is greater than $2$ $\mu$M").  

To demonstrate the applicability of our method , we first study the electrical activity of cardiac cells represented by a hybrid model. By varying parameters we analyze key dynamical properties on multiple cell types, in healthy and disease conditions, and under different input stimuli. We also analyze a hybrid model of the circadian rhythm, and find distinct roles of multiple feedback loops in maintaining oscillatory properties of the dynamics. 

\subsection{Related work}
Mode transitions have been approximated as random events in the literature. In \cite{abate2005stochastic} the dynamics of a hybrid system is approximated by substituting the guards with probabilistic barrier functions.  Our transition probabilities are constructed using similar but simpler considerations. We have done so in order to be able to carry out temporal logic based verification based on simulations. An alternative approach to approximately verifying non-linear hybrid systems is one based on $\delta$-reals \cite{deltareach}. Here one verifies bounded reachability properties that are robust under small  perturbations of the numerical values mentioned in the specification. Since the approximation involved is of a very different kind, it is difficult to compare this line of work with ours. However, it may be fruitful to combine the two approaches to verify a richer set of reachability properties.

The present work may be viewed as an extension of~\cite{palaniappan2013} where a \emph{single} system of ODEs is considered. This  method however, breaks down in the multi-mode hybrid setting and one needs to construct--as we do here--an entirely new machinery. Finally, a wealth of literature is available on the analysis of stochastic automata~\cite{cassandras2010stochastic,blom2006stochastic,julius2009approximation,ballarini2011cosmos}. It will be interesting to explore if these methods can be transported  to our setting.

\section{Hybrid automata}

We fix $n$ real-valued  variables $\{x_i\}_{i=1}^n$ viewed as functions of time $x_i(t)$ with  $t \in \mathbb{R}_{+}$, the set of non-negative reals. A valuation of $\{x_i\}_{i=1}^n$  is  $\bv \in \mathbb{R}^n$ with $\bv(i) \in \mathbb{R}$ representing the value of $x_i$. The language of \emph{guards} is given by: (i) $a < x_i $ and $x_i < b$ are guards where $a, b$ are rationals and $i \in \{1, 2, \ldots, n\}$. (ii) If $g$ and $g'$ are guards then so are $g \wedge g'$ and $g \lor g'$.

$\mathcal{G}$  denotes the set of guards. We define $\bv \models g$ (i.e. $\bv$ satisfies the guard $g$)  via: $\bv \models  a < x_i $ iff $a < \bv(i)$ and similarly for $x_i < b$. The clauses for conjunction and disjunction are standard.   We let $\parallel\!\! g\!\! \parallel = \{ \bv \mid \bv \models g\}$. We note that $\parallel\!\! g \!\!\parallel$ is an open subset of $\Re^n$  for every guard $g$. We will abbreviate  $\parallel\!\! g\!\! \parallel$ as $g$.

\begin{definition}\label{def:SHA}
A hybrid automaton is a tuple
$H=(Q, q_{in},$ $ \{F_q(\bx)\}_{q\in Q}, \mathcal{G}, \ra, \INIT)$,
where
\begin{itemize}
\item $Q$ is a finite set of \emph{modes} and $q_{in} \in Q$ is the \emph{initial} mode.
\item For each $q\in Q$, $d\vx/dt = F_q(\vx)$ is a system of ODEs, where $\vx = (x_1, x_2, \ldots,x_n)$ and $F_q = (f^1_q(\vx), f^2_q(\vx),$ $\ldots, f^n_q(\vx))$. Further,  $f^i_q$ is a $C^1$  function for each $i$.
\item $\ra \subseteq (Q, \mathcal{G}, Q)$ is the mode transition relation.

\item $\INIT = (L_1, U_1) \times  (L_2, U_2) \ldots \times (L_n, U_n)$ is the set of initial states where $L_i < U_i$ and $L_i, U_i$ are rationals.
\end{itemize}
\end{definition}

We have not associated invariant conditions with the modes or reset conditions with the mode transitions. They can be introduced with some additional work.

Fixing a suitable unit time interval $\Delta$, we discretize the time domain as $t = 0, \Delta, 2 \Delta, \ldots$. We assume the states of the system are observed only at these discrete time points. Furthermore, we shall assume that  only a bounded number of mode changes can take place between successive discrete time points. Both in engineered and biological processes this is a reasonable assumption. Given this, we shall in fact assume that $\Delta$ is such that at most one mode change takes place within a $\Delta$ time interval. We note that there can be multiple choices for $\Delta$ that meet this requirement and in practice one must choose this parameter carefully. (Our method can be extended to handle a bounded number of mode transitions in a unit time interval but this will entail notational complications that will obscure the main ideas.) In what follows, for technical convenience we also assume the time scale has been normalized so that $\Delta = 1$. As a result, the discretized set of time points will be $\{0, 1, 2, \ldots\}$.

\subsection{Trajectories}

We have assumed that for every mode $q$, the right hand side of the ODEs, $F_q(\vx)$, is $C^1$ for each component. As a result, for each value $\vv \in \mathbb{R}^n$  and in each mode $q$, the system of ODEs $d\vx/dt = F_q(\vx)$ will have a unique solution  $Z_{q, \vv}(t)$ \cite{hirsch2012differential}. We are also guaranteed that $Z_{q, \vv}(t)$ is Lipschitz and hence measurable \cite{hirsch2012differential}. It will be convenient to work with two sets of functions derived from solutions to the ODE systems.

The (unit interval) \emph{flow} $\Phi_q : (0, 1) \times \mathbb{R}^n \ra \mathbb{R}^n$ is given by $\Phi_q(t, \vv) = Z_{q, \vv}(t)$. $\Phi_q$ will also be Lipschitz. Next we define the parametrized family of functions  $\Phi_{q, t} : \mathbb{R}^n \ra \mathbb{R}^n$ given by $\Phi_{q, t}(\vv) = \Phi_q(t, \vv)$. In addition to being Lipschitz, these functions will be bijective as well. Further, their inverses will also be bijective and Lipschitz.

A (finite) \emph{trajectory} is a sequence
$\tau = (q_0, \vv_0) \, (q_1, \vv_1) \, \ldots (q_k, \vv_k)$ such that for $0 \leq j < k$  the following conditions are satisfied: (i) For $0 \leq j < k$, $q_{j} \stackrel{g_j}{\ra} q_{j +1}$ for some guard $g_j$.
(ii) there exists $t \in (0, 1)$ such that $\Phi_{q_{j}, t}(\vv_{j}) \in g$. Furthermore  $\vv_{j + 1} = \Phi_{q_{j +1}, 1-t}(\Phi_{q_{j}, t}(\vv_{j}))$.

We say that the trajectory  $\tau$ as defined above \emph{starts} from $q_0$ and \emph{ends} in $q_k$. Further, its initial value state is $\vv_0$ and its final value state is $\vv_k$.
We let $TRJ$ denote the set of all finite trajectories that start from the initial mode $q_{in}$ and with an initial value state in $\INIT$.

\section{The Markov chain approximation}\label{sec:construction}
A (finite) path in $H$ is a sequence $\path=q_0q_1 \ldots q_k$ such that for $0 \leq j < k$, there exists a guard $g_j$ such that $q_{j} \stackrel{g_j}{\ra} q_{j + 1}$. We say that this path starts from $q_0$, ends at $q_k$ and is of length $k + 1$.  We let $\paths_H$ denote the set of all finite paths that start from $q_{in}$.

In what follows $\mu$  will denote the standard Lebesgue measure over finite dimensional Euclidean spaces. We will construct $M_H = (\Upsilon, \Rightarrow)$, the Markov chain approximation of $H$ inductively. Each  state in $\Upsilon$ will be of the form $(\rho, X, \textbf{P}_X)$ with $\rho \in \paths_H$, $X$ an open subset of  $\mathbb{R}^n$ of non-zero, finite measure and $\textbf{P}_X$ a probability distribution over $SA(X)$, the $\sigma$-algebra generated by $X$.

We start with $(q_{in}, \INIT, \textbf{P}_{\INIT}) \in \Upsilon$. Clearly,  $\INIT$ is an open set  of non-zero, finite measure  since $\mu(\INIT) = \prod_i (U_i - L_i)$. For technical convenience we shall assume $\textbf{P}_{\INIT}$ to be the uniform probability distribution. In other words, each member of $\INIT$ is an equally likely initial value state. However we can handle other distributions over $\INIT$ as well.
Assume inductively that $(\rho, X, \textbf{P}_X)$ is in $\Upsilon$ with $X$ an open subset of $\mathbb{R}^n$ of
non-zero, finite measure and $\textbf{P}_X$ a probability distribution over $SA(X)$. Suppose $\rho$ ends in $q$ and there are $m$ outgoing transitions $q\move{g_1}q_1,\ldots, q\move{g_m}q_m$ from $q$ in $H$ (Fig. \ref{fig:mc} illustrates this inductive step).

\begin{figure}[htb]
	\centering
	\begin{tikzpicture}[node distance=0.5cm and 0.5cm,shorten >=1pt,auto,>=stealth',
	mynode/.style={ellipse,draw,font=\normalsize,text width=2.3cm,minimum height=0.9cm,minimum width=1.5cm,align=center}, every node/.style={scale=0.7}]
	\node[mynode](init) at (0,0) {$(q_{in},\INIT,\textbf{P}_\INIT)$};
	\node[draw=none](u1) at (-4,-1.1) {};
	\node[draw=none](u2) at (-2,-1.1) {};
	\node[draw=none](d1) at (2,-1.1) {};
	\node[draw=none](d2) at (4,-1.1) {};
	\node[mynode](r) at (0,-1.1) {$(\rho,X,\textbf{P}_X)$};
	\node[mynode](ru1) at (-4,-2.2) {$(\rho q_1,X_1,\textbf{P}_{X_1})$};
	\node[draw=none](ru2) at (-2,-2.2) {$\ldots$};
	\node[mynode](rr) at (0,-2.2) {$(\rho q_j,X_j,\textbf{P}_{X_j})$};
	\node[draw=none](rd2) at (2,-2.2) {$\ldots$};
	\node[mynode](rd1) at (4,-2.2) {$(\rho q_m,X_m,\textbf{P}_{X_m})$};
	
	\draw[->,decorate,decoration={snake,amplitude=.4mm,segment length=2mm,post length=1mm}](init) to node[right]{$\rho$} (r);
	\draw[->,dashed] (init) to (u1);
	\draw[->,dashed] (init) to (u2);
	\draw[->,dashed] (init) to (d1);
	\draw[->,dashed] (init) to (d2);
	\draw[->] (r) to (ru1);
	\draw[->,dashed] (r) to (ru2);
	\draw[->,dashed] (r) to (rr);
	\draw[->] (r) to node[right,rotate=60, xshift=-13pt]{\Huge{$\times$}} (rd1);
	\draw[->,dashed] (r) to (rd2);
	\end{tikzpicture}
	\caption{The Markov chain construction. The edge from the state $(\rho,X,\textbf{P}_X)$ to the state $(\rho q_m,X_m,\textbf{P}_{X_m})$ marked with a `{\large $\times$}' represents the case where $X_m$ has measure 0, and hence the probability of this transition is 0. Thus, $(\rho q_m,X_m,\textbf{P}_{X_m})$  will not be a state of the Markov chain.}
	\label{fig:mc}
\end{figure}
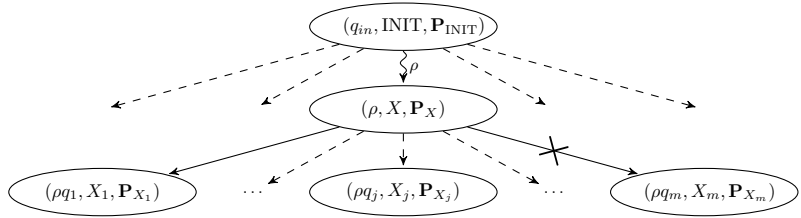

Then for $1 \leq j \leq m$ we define the triples $(\rho q_j, X_j, \textbf{P}_{X_j})$ as follows. In doing so we will assume the required properties
of the objects involved in this construction. We will then establish these properties and thus the soundness of the construction. For convenience, through the remaining parts of this section $j$ will range over $\{1, 2, \ldots, m\}$.

For each $\vv \in X$ and each $j$ we first define the set of time points $\tm_j(\vv) \subseteq (0, 1)$ via
\begin{equation}
\tm_j(\vv) = \{t \ | \Phi_q(t,\vv) \in g_j\}.
\end{equation}
Thus $\tm_j(\vv)$ is the set of time points in $(0, 1)$ at which the guard $g_j$ is satisfied if the system starts from $\vv$ in mode $q$ at time $k$ and evolves according to dynamics of mode $q$ up to time $k + t$.
We next define $X_j$ for each $j$ as
\begin{equation}
X_j = \bigcup_{\vv \in X} \{\Phi_{q_j}(1-t,\Phi_{q}(t,\vv)) \mid t \in \tm_j(\vv)\}.
\end{equation}
Thus $X_j$ is the set of all value states obtained by starting from some $\vv \in X$ at time $k$, evolving up to $k + t$ according to the dynamics $q$, making an instantaneous mode switch to $q_j$ at this time point, and evolving up to time $k + 1$ according to dynamics of mode $q_j$.

To complete the definition of the triples $(\rho q_j, X_j, \textbf{P}_{X_j})$, we first denote by $\Ptm$ the uniform probability distribution over $\tm_j(\vv)$. Our construction can be easily extended to handle other kinds of distributions as well.  We now define the probability distributions $\textbf{P}_{X_j}$ over $SA(X_j)$ as follows.
Suppose $Y$ is a measurable subset of $X_j$. Then
\begin{equation}\label{eqn:measure}
\textbf{P}_{X_j}(Y) = \int_{\vv\in X}\int_{t\in \tm_j(\vv)}{\bf 1}_{(\flow_{q_j}(1-t,\flow_q(t,\vv))\cap Y)} d\Ptm d\textbf{P}_{X}.
\end{equation}

As usual ${\bf 1}_{Z}$ is the indicator function of the set $Z$ while $d\Ptm$ indicates that the inner integration over $\tm_j(\vv)$ is  w.r.t. the (uniform) probability measure $\Ptm$ and $d\textbf{P}_X$ indicates that the outer integration over $X$ is w.r.t. the probability measure $\textbf{P}_X$. Thus $\textbf{P}_{X_j}(Y)$ captures the
probability that the value state $\Phi_{q_j}(1-t,\Phi_{q}(t,\vv))$ lands in $Y \subseteq X_j$ by taking the transition $q \stackrel{g_j}{\ra} q_j$ at some time point in $\tm_j(\vv)$ given that one started with some value state in $X$.

Next we define the triples $((\rho, X, \textbf{P}_X), p_j,  (\rho q_j, X_j,\textbf{P}_{X_j}))$, where $p_j$ is given by

\begin{equation}\label{eqnprob}
p_j = \int_{\vv\in X} \frac{\mu(\tm_j(\vv))}{\sum_{\ell=1}^m
\mu(\tm_{\ell}(\vv))}d\textbf{P}_X.
\end{equation}

Thus $p_j$ captures the probability of taking the mode transition $q \move{g_j} q_j$ when starting from the value states in $X$ and mode $q$. For every $j$ we add the state $(\rho q_j, X_j, \textbf{P}_{X_j})$ to $\Upsilon$ and the triple $((\rho, X, \textbf{P}_X),  p_j,  (\rho q_j, X_j, \textbf{P}_{X_j}))$ to $\Rightarrow$ iff $\mu(X_j) > 0$.

Finally, $(q_{in}, \INIT, \textbf{P}_{\INIT})$ is the initial state of $M_H$.
We can summarize the key properties of our construction as follows (while assuming the associated terminology and notations).

\begin{restatable}[]{theorem}{thmmc}
\label{thm:MCprop}
\begin{enumerate}
\item $\tm_{j}(\vv)$ is an open set of finite measure for each $\vv \in X$ and each $j$.
\item $X_j$ is open and is of finite measure for each $j$.
\item If $(\rho q_j, X_j, \textbf{P}_{X_j}) \in \Upsilon$ then $\mu(X_j) > 0$.
\item $\textbf{P}_{X_j}$ is a probability distribution for each $j$.
\item $M_H = (\Upsilon, \Rightarrow)$ is an infinite state Markov chain whose underlying graph is a finitely
branching tree.
\end{enumerate}

\begin{proof}
	To prove the first part, suppose $t\in \tm_{j}(\vv)$. Then $\flow_q(t,\vv) = \vv' \in g_j$ and $g_j$ is open. Hence $\vv'$ will be contained in an open neighborhood $U$ contained in $g_j$. Since $\flow_q$ is Lipschitz we can pick $U$
	such that $Y' = \flow_q^{-1}(U)$ is an open set containing $(\vv, t)$ with $Y' \subseteq (0, 1) \times X$.  Thus every element of $\tm_{j}(\vv)$ is contained in an open neighborhood in $(0, 1)$ and hence $\tm_{j}(\vv)$ is open.
	
	Using the definition of $X_j$, the fact that $X$ and $\tm_j(\vv)$ are open, and the continuity of the inverses of the flow functions it is easy to observe that  $X_j$ is open. To see that it is of finite measure, by the induction hypothesis, $X$ is open and $\mu(X)$ is finite. Hence $((0,1) \times X)$ is open as well and $\mu((0,1)\times X)$ is finite. Since $\reals^{n+1}$ is second-countable \cite{stephen1970general}, there exists a countable family of disjoint open-intervals $\{I_i\}_{i\geq 1}$ in $\reals^{n+1}$ such that $((0,1)\times X)=\bigcup_i I_i$. Clearly each $I_i$ has a finite measure. By the Lipschitz continuity of $\Phi_q$ we know that there exists a constant $c$ such that $\mu(\Phi_q(I_i)) < c\cdot \mu(I_i)$ for all $i$. We thus have
	\begin{align}
		\mu(\Phi_q((0,1), X)) &\leq \sum_i\mu(\Phi_q(I_i)) \nonumber \\
		&< c\sum_i\mu(I_i) = c\mu((0,1)\times X) < \infty.
	\end{align}
	Therefore  $\Phi_q((0,1), X)$ has a finite measure. By a similar argument we can show that $\Phi_{q_j}((0,1),\Phi_q((0,1), X))$ has a finite measure as well. Since $X_j = \bigcup_t \Phi_{q_j, 1 - t}(\Phi_{q, t}(X) \cap g)\subseteq \Phi_{q_j}((0,1),\Phi_q((0,1), X))$, it must have a finite measure.
	
	The remaining parts of the theorem  follow easily from the definitions and basic measure theory.
\end{proof}
\end{restatable}

\section{Relating the behaviors of {\large{\bf $H$}} and {\large{\bf $M_H$}}}\label{pbltl}
We shall use bounded linear-time temporal logic (BLTL)~\cite{clarke1999model} to specify time bounded properties and use it to relate the behaviors of $H$ and $M_H$. For convenience we shall write $M$ instead of $M_H$ from now on.

We assume a finite set of atomic propositions $AP$ and a valuation function $Kr: Q \ra 2^{AP}$. Formulas of BLTL are defined as: (i) Every atomic proposition as well as the constants $\emph{true}$, $\emph{false}$ are formulas. (ii) If $\psi$, $\psi'$ are formulas then $\lnot \psi$ and $\psi \vee \psi'$  are formulas. (iii) If $\psi$, $\psi'$ are formulas and $\ell$ is a positive integer  then $\psi \U^{\leq \ell}  \psi'$ is a  formula. The derived operators $\F^{\leq\ell}$ and $\G^{\leq \ell}$ are defined as usual: $\F^{\leq\ell}\psi \equiv \emph{true}\U^{\leq\ell}\psi$ and $\G^{\leq\ell}\psi \equiv \neg \F^{\leq\ell}\neg\psi$.

We shall assume through the rest of the paper that the behavior of the system is of interest only up to a maximum time point $K > 0$. This is guided by the fact that given a BLTL formula $\psi$ there is a constant $K_{\psi}$ that depends only on $\psi$ so that it is enough to evaluate an execution trace of length at most $K_{\psi}$ to determine whether $\psi$  is satisfied \cite{TACAS99}. Hence we assume that a sufficiently high $K$ has been chosen to handle the specifications of interest. Having fixed $K$, we denote by $TRJ^{K+1}$ the trajectories of length $K + 1$, and view this set as representing the time bounded non-deterministic behavior of $H$ of interest.

To develop the corresponding notion for $M$, we first define a finite path in $M$ to be  a sequence $\eta_0 \eta_1 \ldots \eta_k$ such that $\eta_j \in \Upsilon$ for $0 \leq j \leq k$. Furthermore for $0 \leq j < k$ there exists $p_j \in (0, 1]$ such that $\eta_j \stackrel{p_j}{\Rightarrow} \eta_{j + 1}$.  Such a path is said to start from $\eta_0$ and its length is $k + 1$. We define $\paths_M$ to be the set of finite paths that start from the initial state of $M$ while $\paths^{K + 1}_M$ is the set of paths in $\paths_M$ of length $K + 1$.

\paragraph*{The trajectory semantics}
Let $\tau = (q_0, \vv_0)$ $(q_1, \vv_1)$ $\ldots$  $(q_k, \vv_k)$ be a finite trajectory, $\psi$ a BLTL formula and $0\leq j \leq K$. Then $\tau, j \models_H \psi$  is defined via:
\begin{itemize}
\item $\tau, j \models_H  A$ iff $A \in Kr(q_j)$, where $A$ is an atomic proposition.
\item $\lnot$ and  $\vee$ are interpreted in the usual way.
\item $\tau,j \models_H \psi \U^{\leq {\ell}} \psi'$ iff there exists $j'$ such that $j' \leq \ell$ and $j + j' \leq k$ and $\tau, (j + j') \models_H\psi'$. Further, $\tau,(j + j'') \models_H \psi$ for every $0 \leq j'' < j'$.
\end{itemize}

We now define $models_H(\psi) \subseteq TRJ^{K + 1}$ via: $\tau \in models_H(\psi)$ iff $\tau, 0 \models_H \psi$. We say that $H$ \emph{meets the specification} $\psi$ -denoted $H \models \psi$- iff $models_H(\psi) = TRJ^{K + 1}$.

\paragraph*{The Markov chain semantics}
Let $\pi= \eta_0 \eta_1 \ldots \eta_k$ be a path in $M$ with $\eta_j = (\rho q_j, X_j, \textbf{P}_{X_j})$ for $0 \leq j \leq k$.
Let $\psi$ be a BLTL formula and $0 \leq j \leq k$.  Then  $\pi, j \models_M \psi$ is given by:
\begin{itemize}
\item $\pi, j \models_M  A$ iff $A \in Kr(q_j)$, where $A$ is an atomic proposition.
\item The remaining clauses are defined just as in the case of $\models_H$. 
\end{itemize}

Now we define $models_M(\psi) \subseteq \paths^{K + 1}_M$ via: $\pi \in models_M(\psi)$ iff $\pi, 0 \models_M \psi$.
We can now define the probability of satisfaction of a formula in $M$. Let $\pi = \eta_0 \eta_1 \ldots \eta_K$ be in $\paths^{K + 1}_M$. Then $\Pr(\pi) = \prod_{0 \leq \ell < K} p_{\ell}$,
where $\eta_{\ell} \stackrel{p_{\ell}}{\Rightarrow} \eta_{\ell + 1}$ for $0 \leq \ell < K$. This leads to
\begin{equation*}
\Pr(models_M(\psi)) = \sum_{\pi \in models_M(\psi)} \Pr(\pi).
\end{equation*}
We write $\mc \models \psi$ to denote $\Pr(models_M(\psi)) = 1$

For $p \in [0, 1]$ we write as usual $\Pr_{\geq p}(\psi)$ instead of \\
$\Pr(models_M(\psi))$ $ \geq p$.
We note that $Pr(\pi) > 0$ for every $\pi \in models_M(\psi)$. Furthermore $\sum_{\pi \in models_M(\psi)} \Pr(\pi) \leq 1$. Hence $\Pr_{\geq 1}(\psi)$ iff $ models_M(\psi) = \paths^{K + 1}_M$ iff $M \models \psi$.
\subsection{The correspondence result}
We wish to show that $H$ meets the specification $\psi$ iff $\Pr_{\geq 1}(\psi)$. To this end let $\pi= \eta_0 \eta_1 \ldots \eta_k$ be a path in $M$ with $\eta_j = (q_0q_1\ldots q_j, X_j, \textbf{P}_{X_j})$ for $0 \leq j \leq k$ and let $\tau = (q'_0, \vv_0)$ $(q'_1, \vv_1) \ldots(q'_{k'}, \vv_{k'})$ be a trajectory. Then we say that $\pi$ and $\tau$ are \emph{compatible} iff $k = k'$ and $q_j = q'_j$ and $\vv_j \in X_j$ for $0 \leq j \leq k$. The following three observations based on this notion will easily lead to the main result.

\begin{restatable}[]{lemma}{lempbltl}
\label{lemmapbltl}
\begin{enumerate}
\item Suppose the path $\pi= \eta_0 \eta_1 \ldots \eta_k$ in $M$ and the trajectory $\tau = (q_0, \vv_0)$ $(q_1, \vv_1) \ldots(q_k, \vv_k)$ are compatible. Let $0 \leq j \leq k$ and $\psi$ be a BLTL formula. Then $\pi, j \models_M \psi$ iff $\tau, j \models_H \psi$.
\item Suppose $\pi$ is a path in M. Then there exists a trajectory $\tau$  such that $\pi$ and $\tau$ are compatible. Furthermore if $\pi \in \paths_M$ then $\tau \in TRJ$.
\item Suppose $\tau$ is a trajectory. Then there exists a path $\pi$ in $M$ such that $\tau$ and $\pi$ are compatible. Furthermore if $\tau \in TRJ$ then $\pi \in \paths_M$.
\end{enumerate}

\begin{proof}
	To prove the first part we note that if $A$ is an atomic proposition then $\pi, j \models_M A$ iff $A \in Kr(q_j)$ iff $\tau, j \models_H A$. We next note that the suffix of length $m$ of $\pi$ will be compatible with the suffix of length $m$ of $\tau$ whenever $\pi$ and $\tau$ are compatible. The result now follows at once by structural induction on $\psi$.
	
	To show the second part let $\pi = \eta_0 \eta_1 \ldots \eta_k$ be a path in $M$ with $\eta_j = (q_0q_1\ldots q_j, X_j, \textbf{P}_{X_j})$ for $0 \leq j \leq k$. Clearly $X_j$ is non-empty for $0 \leq j \leq k$ since $\eta_j \in \Upsilon$ implies $\mu(X_j) > 0$.
	We proceed by induction on $k$. If $k = 0$ then we can pick $\vv_0 \in X_0$ and the trajectory $(q_0, \vv_0)$ will be compatible with $\tau$. So assume $k > 0$. Then by the induction hypothesis there exists a trajectory $(q_1, \vv_1) (q_2, \vv_2) \ldots (q_k, \vv_k)$ which is compatible with the path $\eta_1 \eta_2 \ldots \eta_k$. Let $q_0 \stackrel{g}{\ra} q_1$. Since $\vv_1 \in X_1$ there must exist $\vv_0$ in $X_0$ and $t \in (0, 1)$ such that $\Phi_{q_0, t}(\vv_0) \in g$ and $\vv_1 = \Phi_{q_1, 1-t}(\Phi_{q_0, t}(\vv_0))$. Clearly $\vv_0 \vv_1 \ldots \vv_k$ is a trajectory that is compatible with $\pi$.
	The fact that $\tau \in TRJ$ if $\pi \in \paths_M$ follows from the definition of compatibility.
	
	To prove the third part let $\tau = (q_0, \vv_0)$ $(q_1, \vv_1) \ldots(q_k, \vv_k) \in TRJ$. Again we proceed by induction
	on $k$. Suppose $k = 0$. Then $(q_{in}, \INIT, \textbf{P}_{\INIT})$ is in $\paths_M$ which is compatible with $\tau$. So suppose $k > 0$. Then by the induction hypothesis there exits $\pi' = \eta_0 \eta_1 \ldots \eta_{k - 1}$ such that $\pi'$ is compatible with $\tau' = (q_0, \vv_0) (q_1, \vv_1) \ldots (q_{k - 1}, \vv_{k -1})$. Let $q_{k -1} \stackrel{g}{\ra} q_k$. Since $X_{k -1}$ is open there exists an open neighborhood  $Y \subseteq X_{k -1}$ that contains $\vv_{k - 1}$. But then both $\Phi^{-1}_{q_{k -1}}$ and $\Phi^{-1}_{q_k}$ are continuous bijections. Thus $\Phi_{q_{k-1},t}(Y)$ is open and $\Phi_{q_{k-1},t}(Y)\cap g$ should be open and non-empty (since $g$ is open and $(q_k,\vv_k)$ is part of the trajectory). Hence $Y'=\bigcup_{t \in (0, 1)} \Phi_{q_{k}, 1 - t}(\Phi_{q_{k - 1},t}(Y)\cap g)$ is a non-empty open set with a positive measure.
	Hence there will be a state of the form $\eta_k = (\rho_k, X_k, \textbf{P}_{X_k})$ in $\Upsilon$ with $Y' \subseteq X_k$ and $\eta_{k - 1} \stackrel{p}{\Rightarrow} \eta_k$ for some $p \in (0, 1]$. Clearly $\pi = \pi' \eta_k \in \paths_M$  and is compatible with $\tau$. Again the fact that $\pi \in \paths_M$ if $\tau \in TRJ$ follows from the definition of compatibility.
\end{proof}
\end{restatable}

\begin{restatable}[]{theorem}{thmmain}
\label{thm:main}
$H \models \psi$ iff $M \models \psi$.
\begin{proof}
Suppose $H$ does not meet the specification $\psi$. Then there exists $\tau \in TRJ^{K+1}$ such
that $\tau, 0 \not\models_H \psi$. By the third part of Lemma~\ref{lemmapbltl} there exists $\pi \in \paths^{K+ 1}_M$
which is compatible with $\tau$. By the first part of Lemma~\ref{lemmapbltl} we then have $\pi \notin models_M(\psi)$ which leads to $Pr_{< 1}(\psi)$.

Next suppose that $Pr_{< 1}(\psi)$. Then there exists $\pi \in \paths^{K + 1}$ such that $\pi, 0 \not\models_M \psi$.
By the second part of Lemma~\ref{lemmapbltl} there exists $\tau \in TRJ^{K + 1}$ which is compatible with $\pi$. By the first part of Lemma~\ref{lemmapbltl} this implies $\tau, 0 \not\models_H \psi$ and this in turn implies that $H$ does not meet the specification $\psi$. \qed
\end{proof}
\end{restatable}

\section{The SMC procedure}\label{sec:smc}
To verify whether $H$ meets the specification $\psi$, we solve the equivalent problem whether $\prob_{\ge 1}(\psi)$ on $M$.
However as discussed in Section \ref{sec:intro}, $M$ cannot be constructed explicitly since both its structure and transition probabilities, defined in terms of the solutions to the ODEs, will not be available.
Therefore we shall use randomly generated trajectories to sample the paths of $M$ and formulate a sequential hypothesis test to decide with bounded error rate whether $\prob_{\ge 1}(\psi)$ holds.
Algorithm \ref{algtrajsim} describes our trajectory sampling procedure.
\vspace{-0.2in}
\begin{algorithm}[H]
\scriptsize
  \renewcommand{\thealgorithm}{1}
	\caption{Trajectory simulation}
	\label{algtrajsim}
	Input: Hybrid automaton $H=(Q, q_{in}, \{F_q(\bx)\}_{q\in Q}, \mathcal{G}, \ra, \INIT)$, maximum time step $K$.
	
	Output: Trajectory $\tau$
	\begin{algorithmic}[1]
		\State Sample $\vv_0$ from $\INIT$ uniformly, set $q_0:=q_{in}$ and $\tau:=(q_0,\vv_0)$.
		\For{$k:=1 \ldots K$}
		\State Generate time points  $T:=\{t_1,\ldots,t_J\}$ uniformly in $(0,1)$.
		\State Simulate $\vv^j := \Phi_{q_{k-1}}(t_j,\vv_{k-1})$, for $j\in\{1,\ldots,J\}$
		\State Let $\widehat{\tm}_j:=\{t \in T: \vv^j \in g_j\}$ be the time points where $g_j$ is enabled.
		\State Pick $g_{\ell}$ randomly according to probabilities $\{p_j := |\widehat{\tm}_j| / \sum_{i=1}^m |\widehat{\tm}_i|\}$.
		\State Pick $t_{\ell}$ uniformly at random from $\widehat{\tm}_{\ell}$.
		\State Simulate $\vv':=\Phi_{q'}(1-t_{\ell},\vv^{\ell})$, where $q'$ is the target of $g_{\ell}$.
		\State Set $q_k := q'$, $\vv_k := \vv'$, and extend $\tau := (q_0,\vv_0)\ldots(q_k,\vv_k)$.
		\EndFor
		\State \textbf{return} $\tau$
	\end{algorithmic}
\end{algorithm}
\vspace{-0.2in}

 We now show that the trajectory generation algorithm (Algorithm \ref{algtrajsim}) generates a trajectory in $TRJ^{K + 1}$ whose induced paths in $M$ are being sampled according to the underlying probabilities.
According to Algorithm \ref{algtrajsim}, the probability of picking guard $g_j$ for a trajectory starting at $\vv\in X$ is defined as $|\widehat{\tm}_j| / \sum_{i=1}^m |\widehat{\tm}_i|$, which, by the law of large numbers tends to
\begin{equation}\label{eqn:probguard}
	p_j(\vv) := \frac{\mu(\tm_j(\vv))}{\sum_{i=1}^m \mu(\tm_i(\vv))}
\end{equation}
as $J$ tends to $\infty$.

Now if $\vv$ is randomly sampled according to $\textbf{P}_X$, then the probability of picking guard $j$  can be expressed as the expected value of $p_j(\vv)$ under $\vv\sim \textbf{P}_X$ as
\begin{equation}
	\expect_{\vv\sim \textbf{P}_X}[p_j(\vv)] = \int_{\vv\in X} p_j(\vv)d\textbf{P}_X = \int_{\vv\in X} \frac{\mu(\tm_j(\vv))}{\sum_{i=1}^m \mu(\tm_i(\vv))}d\textbf{P}_X,
\end{equation}
which by \eqref{eqnprob} is equal to $p_j$, the corresponding transition probability in the Markov chain.

\begin{proof}
	Clearly it suffices to show that for a measurable subset $Y\subseteq X_j$, $\prob(\vv'\in Y)=\textbf{P}_{X_j}(Y)$. We start with
	\begin{equation*}
		\prob(\vv'\in Y\ |\ \vv) = \int_{t\in \tm_{j}(\vv)}{\bf 1}_{(\flow_{q_j}(1-t,\flow_q(t,\vv))\cap Y)}d\Ptm.
	\end{equation*}
	Integrating now over all possible choices of $\vv$ with respect to $\textbf{P}_X$ we have
	\begin{equation*}
		\prob(\vv'\in Y) = \int_{\vv\in X}\prob(\vv'\in Y\ |\ \vv)d\textbf{P}_X.
	\end{equation*}
	From \eqref{eqn:measure} it follows that  $\prob(\vv'\in Y) = \textbf{P}_{X_j}(Y)$ with $\vv\sim \textbf{P}_X$ and $t\sim \Ptm$.
\end{proof}

Whether the generated trajectory of length $K + 1$ (and hence the corresponding path of $M$) is a model of  $\psi$ can be determined using a standard BLTL model checker \cite{clarke1999model}. In fact this can be done on the fly which will often  avoid generating the whole trajectory. Based on this, we can test whether  $\prob_{\ge 1}(\psi)$ on $M$  by testing the following alternative pair of hypotheses:
$H_0:  \prob_{\ge 1}(\psi)$ and $H_1:  \prob_{<1-\delta}(\psi)$,
where $0 < \delta < 1$ is a parameter chosen by the user marking the interval $[1-\delta,1)$ as an indifference region in which accepting either hypothesis is fine. In our setting, whenever we encounter a sample (i.e. a randomly generated trajectory) that does not satisfy $\psi$, we can reject $H_0$ and accept $H_1$. Thus we only have to deal with false positives (when $H_0$ is accepted while $H_1$ happens to be true).

  This leads to Algorithm \ref{alghyptest} that repeatedly generates a random trajectory (using Algorithm \ref{algtrajsim}), and decides after a finite number of tries  between $H_0$ and $H_1$. For doing so we also fix a user-defined false positive rate $\alpha$.
  \vspace{-0.2in}
\begin{algorithm}[H]
\scriptsize
\renewcommand{\thealgorithm}{2}
\caption{Sequential hypothesis test}
\label{alghyptest}
Input: Markov chain $M$, BLTL property $\psi$, indifference parameter $\delta$, false positive bound $\alpha$.

Output: $H_0$ or $H_1$.
	\begin{algorithmic}[1]
		\State Set $N := \lceil \log \alpha/\log (1-\delta) \rceil$
		\For{$i:=1 \ldots N$}
			\State Generate a random trajectory $\tau$ using Algorithm \ref{algtrajsim}
			\State \textbf{if} {$\tau, 0 \models^H \psi$} \textbf{ then } Continue
			\State \textbf{else return} $H_1$
		\EndFor
		\State \textbf{return} $H_0$
	\end{algorithmic}
\end{algorithm}
\vspace{-0.3in}
The accuracy of Algorithm \ref{alghyptest} is captured by the next result.

\begin{restatable}[]{theorem}{thmaccuracy}
  \label{thm:smc} The probability of choosing $H_1$ when $H_0$ is true (false negative) is $0$. Further, suppose $N \ge \log\alpha/\log (1-\delta)$. Then the probability of choosing $H_0$ when $H_1$ is true (false positive) is no more than $\alpha$.
  \end{restatable}
\begin{proof}
As observed earlier the first part is obvious. To prove the second part, if $H_1$ is true, then we know that $\prob_{<1-\delta}(\psi)$. The probability of $N$ sampled trajectories all satisfying $\psi$ (and thus returning $H_0$, a false positive) is at most $(1-\delta)^N$. Therefore we have $\alpha \le (1-\delta)^N$, leading to  $N \ge \log\alpha/\log (1-\delta)$. \qed
\end{proof}

Hence we use $N := \lceil \log\alpha / \log (1-\delta) \rceil$ to set the sample size. For example for $\delta = 0.01$ and $\alpha = 0.01$ we get $N = 459$ while for $\delta = 0.001$ and $\alpha = 0.01$ we get $N = 4603$.

\section{Quantitative specifications}\label{quant}
To specify quantitative properties we fix a finite set of atomic propositions $AP_{qt}$ of the form $\langle x_i < c \rangle$ or $\langle x_i > c \rangle$ where $c$ is a rational constant. In what follows we shall assume for convenience that all the atomic propositions that we encounter are members of $AP_{qt}$. It will be straightforward to extend our arguments to include qualitative atomic propositions as well.

We partition $\Re^n$ into hypercubes according to the constants mentioned in the quantitative atomic propositions in $AP_{qt}$. (Actually one could just focus on the members of $AP_{qt}$ that appear in a given specification but we wish to deal with specifications later). Accordingly, define $C_i$ to be the set of rational constants so that $c \in C_i$ iff an atomic proposition of the form $\langle x_i < c \rangle$ or $\langle x_i > c \rangle$ appears in $AP_{qt}$. We next define for each dimension $i$ the set of intervals
\[ \ints_i = \{(-\infty,c^1_i), \{c_i^1\}, (c_i^1,c_i^2), \{c_i^2\}, \ldots (c_i^m,+\infty)\}\]
where $C_i = \{c_i^1 < c_i^2 < \ldots <  c_i^m \}$. In case $C_i = \emptyset$ we set $\ints_i = \{(-\infty, +\infty)\}$.

This leads to the set of hypercubes $\hyper$ given by $\hyper = \{\prod_iI_i\ |\ I_i\in \ints_i\}$.
Clearly $\hyper$ is a partition of $\Re^n$. The states of the Markov chain $\mc_{qt}$ we wish to define as the approximation of $H$ will be the states of $\mc$ defined previously but now refined using $\hyper$. More precisely  we define $\mc_{qt}= (\states_{qt}, \edges_{qt})$ inductively as follows: $\epsilon \in \states_{qt}$ and it is the initial state of $\mc_{qt}$. Every other state in $\states_{qt}$ will be of the form $(\path, X,\hcube,\mathbf{P}_{X})$ where $\path$ is a path in H, $X$ is an open subset of $\reals^n$ of finite non-zero measure, $\hcube\in \hyper$ and $\mathbf{P}_{X}$ is a probability distribution over $X$. Furthermore $X \subseteq \hcube$.

\subsection{The two semantics}
For interpreting $BLTL$ formulas over $\mc_{qt}$ it will be convenient to assume the following syntax in which negation is immediately followed by a quantitative atomic proposition:
\[A \, | \, \lnot A \, | \, \varphi_1 \lor \varphi_2 \, | \, \varphi_1 \wedge \varphi_2 \, | \, G^{\leq k} \varphi \, | \, F^{\leq k} \varphi \, | \, \varphi_1 \U^{\leq k} \varphi_2.\]

Clearly, every BLTL formula can be transformed into an equivalent formula that has the above syntax. This can be achieved by pushing negation inwards using equivalences such as $\lnot(\varphi_1 \lor \varphi_2) \equiv \lnot \varphi_1 \wedge \lnot \varphi_2$, $\lnot G^{\leq k} \varphi \equiv F^{\leq k} \lnot \varphi$, $\lnot (\varphi_1 \U^{\leq k}$ $\varphi_2) \equiv G^{\leq k} \lnot \varphi_2 \lor (\lnot \varphi_2 U^{\leq k} (\lnot \varphi_1 \wedge \lnot \varphi_2))$ etc.

The trajectory semantics is defined along previous lines but the atomic propositions are handled as follows. Let $\tau = (q_0, \vv_0)$ $(q_1, \vv_1)$ $\ldots$  $(q_k, \vv_k)$ be a finite trajectory and $0\leq \ell \leq k$. Then $\tau, \ell \models_{H, qt} \langle x_i < c \rangle$ iff
$\vv_{\ell}(i) < c$. On the other hand $\tau, \ell \models_{H, qt} \lnot \langle x_i < c \rangle$ iff $\tau, \ell \not \models_H \langle x_i < c \rangle$. The clauses for the other cases are defined in the obvious way. As before $\tau$ is a (trajectory) model of $\psi$ iff $\tau \in TRJ^{K+1}$ and $\tau, 0 \models_{H, qt} \psi$.

To interpret BLTL formulas over $\mc_{qt}$, let $\pi= \eta_0 \eta_1 \ldots \eta_k$ be a path in $\mc_{qt}$ with $\eta_0=\epsilon$ and $\eta_{\ell} = (\rho q_{\ell}, X_{\ell}, \hcube_{\ell}, \mathbf{P}_{X_{\ell}})$ for $0 < \ell \leq k$.
Let $\psi$ be a BLTL formula and $0 < \ell \leq k$.  Then  $\pi, \ell \models_{qt} \psi$ is given by:
\begin{itemize}
\item $\pi,  \ell \models_{qt}  \langle x_i < c \rangle$ iff there exists $\vv \in X_{\ell}$ such that $\vv(i) < c$.
\item $\pi,  \ell \models_{qt}  \lnot \langle x_i < c \rangle$ iff there exists $\vv \in X_{\ell}$ such that $\vv(i) \geq c$.
\item The remaining clauses are defined in the obvious way.
\end{itemize}
For $\vv \in \Re^n$ let $\vv \models A$ denote the fact that $\vv(i) < c$ in case $A = \langle x_i < c \rangle$ and $\vv(i) > c$ in case $A = \langle x_i > c \rangle$. Next suppose $(\rho, X, \hcube, \mathbf{P}_{X})$ is a state of $\mc_{qt}$ and $A \in AP_{qt}$. Then $X \subseteq \hcube$ by construction. Furthermore it is easy to check that $\vv \models A$ for every $\vv \in \hcube$ or $\vv \not\models A$ for every $\vv \in \hcube$. Thus the semantics defined above will be consistent in the sense it will be the case that either $\pi,  \ell \models_{qt} A$ or $\pi, \ell \models_{qt} \lnot A$ but not both.

Let $\mathcal{B}$ be the set of paths of length $K + 2$ that start from the initial state of $\mc_{qt}$. Now we define $models_{qt}(\psi) \subseteq \mathcal{B}$ via: $\pi \in models_{qt}(\psi)$ iff $\pi, 1 \models_{qt} \psi$.
We can now define the probability of satisfaction of a formula in $\mc_{qt}$. Let $\pi = \eta_0 \eta_1 \ldots \eta_{K+1}\in \mathcal{B}$. Then $\Pr(\pi) = \prod_{0 \leq \ell < K} p_{\ell}$,
where $\eta_{\ell} \stackrel{p_{\ell}}{\Rightarrow} \eta_{\ell + 1}$ for $0 \leq \ell < K+1$. This leads to
\begin{equation*}
\Pr(models_{qt}(\psi)) = \sum_{\pi \in models_{qt}(\psi)} \Pr(\pi).
\end{equation*}

We let $\mc_{qt} \models \psi$ denote the fact $\Pr(models_{qt}(\psi)) = 1$.

\subsection{The correspondence result}
We shall relate the behavior of $H$ to that $\mc_{qt}$ using the notion of \emph{robust} trajectories.  To start with, for $\vv \in \Re^{n}$ we let $hc(\vv)$ be the hypercube $\hcube$ in $\mathcal{H}$ such that $\vv \in \hcube$. Since $\mathcal{H}$ is a partition of $\Re^n$ we have that $hc(\vv)$ exists and is unique. In what follows we let $\ell$ range over $\{0, 1, \ldots, K\}$. We now define the equivalence relation $\approx \subseteq TRJ^{K+1}$ as follows: Let $\tau, \tau' \in TRJ^{K+1}$ with $\tau(\ell) = (q_{\ell}, \vv_{\ell})$ and $\tau'(\ell) = (q'_{\ell}, \vv'_{\ell})$. Then $\tau \approx \tau'$ iff $q_{\ell} = q'_{\ell}$ and  $hc(\vv_{\ell}) = hc(\vv'_{\ell})$ for each $\ell$. We let $[\tau]$ denote the $\approx$-equivalence class containing $\tau$.

Next suppose $\tau \in TRJ^{K+1}$ with $\tau(\ell) = (q_{\ell}, \vv_{\ell})$. Let $\mathcal{Q}(\tau,\ell)=q_\ell$ and $\mathcal{V}(\tau, \ell) = \vv_{\ell}$. Define $[\tau](\ell) = \{\mathcal{V}(\tau', \ell) \mid \tau' \in [\tau]\}$. It is easy to verify that $[\tau](\ell)$ is a measurable set (but perhaps with measure $0$) for each $\ell$.

The trajectory $\tau \in TRJ^{K+1}$ is said to be \emph{robust} iff $\mu([\tau](\ell)) > 0$ for every $\ell$. We will say that $H$ robustly satisfies the specification $\psi$-and this is denoted by $H \models_R \psi$ iff $\tau, 0 \models_H \psi$ for every robust trajectory $\tau$ in $TRJ^{K+1}$. It is now straightforward to show (along the lines of the proof of \ref{thm:main}) show:

\begin{restatable}[]{theorem}{thmqt}
\label{thmqt}
$H \models_R \psi$ iff $\mc_{qt} \models \psi$.
\end{restatable}

First the following properties of  the  Markov chain $M_{qt}$ can easily be proved along the lines of the proof of Theorem \ref{thm:MCprop}.

\begin{lemma}\label{lem:MCqt}
\begin{enumerate}
\item $X_j^\hcube$ is open and is of finite measure for each $j$ and each $\hcube\in \hyper$.
\item If $(\rho q_j, X_j^\hcube,\hcube, \textbf{P}_{X^\hcube_j}) \in \Upsilon_{qt}$ then $\mu(X_j^\hcube) > 0$.
\item $\textbf{P}_{X^\hcube_j}$ is a probability distribution for each $j$ and each $\hcube\in \hyper$.
\item $M_{qt} = (\Upsilon_{qt}, \Rightarrow_{qt})$ is an infinite state Markov chain whose underlying graph is a finitely
branching tree.
\end{enumerate}
\end{lemma}

We wish to show that for quantitative specifications, $H$ robustly satisfies a BLTL specification $\psi$ if and only if $M_{qt}$ satisfies $\psi$ with probability 1. We begin with:

\begin{lemma}\label{lem:trajclose}
Let $\traj = (q_0,\vv_0),(q_1,\vv_1),\ldots (q_{K},\vv_{K}) \in TRJ^{K+1}$. Then the following statements are equivalent.
\begin{enumerate}
\item $\traj$ is robust.
\item There exist open sets of non-zero measure $O_j$ and $\hcube_j \in \hyper$ such that $\vv_j \in O_j \subseteq [\trj][j] \subseteq \hcube_j$ for $0 \leq j \leq K$.
\item $\vv_j(i) \notin C_i$ for every $j \in \{0, 1, \leq K\}$ and every $i \in \{1, 2, \ldots, n\}$.
\end{enumerate}
\end{lemma}

\begin{proof}
In what follows we let $j$ range over $\{0, 1, \ldots, K\}$.
Suppose $\traj$ is robust. Let $hc(\vv_j) = \hcube_j$ for each $j$. By the definition of $\approx$, we have   $\vv_j \in [\traj](j) \subseteq \hcube_j$ for each $j$. Since $\mu([\traj](j)) > 0$ we have $\mu(\hcube_j) > 0$ for each $j$. This implies that $\hcube_j(i)$ is a finite  open interval for $1 \leq i \leq n$. But then $[\traj](j) \subseteq \hcube_j$ and $\mu([\traj](j)) > 0$ now together imply that there exists a non-empty open set $O_j$ of finite measure such that $\vv_j \in O_j \subseteq [\traj](j)$ for each $j$. Thus (1) implies (2).

 Next suppose part (2) of the lemma holds. Then $\mu([\traj](j)) > 0$ for each $j$. Thus $\traj$ is robust and we have (2) implies (1).

 To show that (2) implies (3) assume that $\vv_j(i) \in C_i$ for some $j$ and $i$. Then $\mu(hc(\vv_j)) = 0$ . We need to find $\hcube_j$ and an open set of non-zero measure such that $\vv_j \in O_j \subset [\traj](j) \subseteq \hcube_j$. This implies  $hc(\vv_j) = \hcube_j$. But then $\mu(\hcube_j) = 0$ implies there can not exist an open set $O_j$ of \emph{non-zero measure} satisfying $\vv_j \in O_j \subseteq \hcube_j$. Hence (2) can not hold and this shows (2) implies (3).

 Next suppose (3) holds. Let $\hcube_j = hc(\vv_j)$ for each $j$. Then (3) implies $\mu(\hcube_j) > 0$ for each $j$.
 Let $\traj^{(j)}$ be the $j$-length prefix of $\traj$ for each $j$.

 Since $\INIT$ is open $O_0 = \INIT \cap \hcube_0$ is open. It is non-empty since $\vv_0 \in O_0$ and hence has non-zero measure. Furthermore $[\traj^{(0)}](0) = O_0$. We now have $\vv_0 \in O_0 \subseteq [\traj^{(0)}](0) \subseteq \hcube_0$. Assume inductively $0 < j < K$ and for $0 \leq k \leq j$ there exist open sets $O_k$ of non-zero measure such that $\vv_k \in O_k \subseteq [\traj^{(j)}](k) \subseteq \hcube_k$.

 Since $\traj$ is a trajectory there exist $g_j$ and $t_j\in(0,1)$ such that $q_j\move{g_j}q_{j+1}$ and $\Phi_{q_j,t_j}(\vv_j)\in g_j$ and $\vv_{j+1}=\Phi_{q_{j+1},1-t_j}(\Phi_{q_{j},t_j}(\vv_j))$. Let $Y_j = [\traj^{(j)}](j)$ and $Y'_{j+1}=\bigcup_{\vv \in Y_j} \{\Phi_{q_{j+1}, 1- t}(\Phi_{q_j, t}(\vv)) | t \in\tm(\vv)\}$ where $\tm(\vv) = \{ t | \Phi_{q_j, t}(\vv) \in g\}$. Clearly  $[\traj^{(j+1)}](j+1) = Y'_{j+1} \cap \hcube_{j+1}$. Next define $O'_{j+1} = \Phi_{q_{j+1},1-t_j}(\Phi_{q_{j},t_j}(O_j))$. Since both $\Phi^{-1}_{q_j,1-t_j}$ and $\Phi^{-1}_{q_{j},t_j}$ are continuous bijections, $O'_{j+1}$ is an open set and $\vv_{j+1}\in O'_{j+1}$. Let $O_{j+1} = O'_{j+1} \cap \hcube_{j+1}$. Since $\vv_{j+1} \in \hcube_{j+1}$ and $\hcube_{j+1}$ is open we have $O_{j+1}$ is open and non-empty and hence with non-zero measure. Further $O_{j+1} \subseteq [\traj^{(j+1)}](j+1) \subseteq \hcube_{j+1}$. This establishes the induction hypothesis and hence (3) implies (2).
 \end{proof}

We define the notion of {\em compatibility} as before. Let $\pathmc = \state_0\state_1\ldots\state_{k}$ be a path in $\mc_{qt}$ with $\state_j = (q_0q_1\ldots q_{j-1},X_j^{\hcube_j},\hcube_j,\prob_{X_j^{\hcube_j}})$ for $0 < j\leq k$,\ and  $\state_0=\epsilon$. Let $\traj = (q'_1,\vv_1)(q'_2,\vv_1)\ldots (q'_{k'},\vv_{k'})$ be a trajectory. Then we say that $\pathmc$ and $\traj$ are {\it compatible} iff $k=k'$ and for  $1\leq j\leq k$, $q_j = q'_j$ and $\vv_j \in X_{j}^{\hcube_{j}}$. As it will turn out, if $\traj$ and $\pathmc$ are compatible then $\traj$ will be robust.

In what follows we shall assume that our BLTL specifications involve only quantitative atomic propositions in $AP_{qt}$ and the formulas obey the syntax in which negation is immediately followed by an atomic proposition. Further the semantic notions $\models_H$ and $\models_{\mc_{qt}}$ (abbreviated as $\models_{qt}$)  are defined in the expected way.

\begin{lemma}\label{lem:properties}
\begin{enumerate}
\item Suppose the trajectory \linebreak $\traj = (q_1, \vv_1) (q_2, \vv_1) \ldots(q_k, \vv_k) \in \TRJ$  and  the path $\pathmc = \state_0 \state_1 \ldots \state_{k}$  in $\mc_{qt}$ with $\state_0=\epsilon$ are compatible. Let  $\psi$ be a BLTL specification and $j \in \{1, \ldots k\}$. Then $\traj, j \models_H \psi$ iff $\pathmc, j \models_{qt} \psi$.
\item Suppose $\pathmc$ is a path in $\mc_{qt}$ starting from $\epsilon$. Then there exists a robust trajectory $\traj$ in $TRJ$  such that $\pathmc$ and $\traj$ are compatible.
\item Suppose $\traj$ is a robust trajectory in $TRJ$. Then there exists a path $\pathmc$ in $\mc_{qt}$ starting from $\epsilon$ such that $\traj$ and $\pathmc$ are compatible.
\end{enumerate}
\end{lemma}

\begin{proof}
\begin{enumerate}
\item From the definitions it follows that if $A \in AP_{qt}$ and $\hcube \in \hyper$ then $\vv \models A$ for every $\vv \in \hcube$ or $\vv \models \neg A$ for every $\vv \in \hcube$ but not both. Since $\vv_j \in \hcube_j$ we then have $\traj, j \models_H A$ iff $\pathmc, j \models_{qt} A$ and $\traj, j \models_H \neg A$ iff $\pathmc, j \models_{qt} \neg A$ for every atomic proposition. The remaining cases now follow easily by structural induction on $\psi$.

\item Let $\pathmc = \state_0 \state_1 \ldots \state_{k}$  in $\mc_{qt}$ with $\state_0=\epsilon$ and $\state_j = (q_0q_1\ldots q_{j-1},X_j^{\hcube_j},\hcube_j,\prob_{X_j^{\hcube_j}})$ for $0 < j\leq k$. For notational convenience we will write $X_j$ instead of $X_j^{\hcube_j}$.

    Since $\mu(X_k) > 0$ we can fix $\vv_k \in X_k$. Further $\hcube_k$ being a product of open intervals in $\Re$ with $X_k \subseteq \hcube_k$,  we can find an open set $O_k$ of non-zero measure such that $\vv_k \in O_k \subseteq X_k$. Thus we have $\vv_k \in O_k \subseteq X_k \subseteq \hcube_k$. From the construction of $\mc_{qt}$ it follows there exists $q_{k-1} \move{g}q_k$ and $\tm(\vv)  \subseteq (0, 1)$ for each $\vv \in X_k$ such that $\Phi^{-1}_{q_k, 1-t}(\vv) \in g$ for every $t \in \tm(\vv)$. Let $Y_{k-1}= \bigcup_{\vv\in X_k} \{\Phi^{-1}_{q_{k-1}, t}(\Phi^{-1}_{q_k, 1-t}(\vv)) \ |\ t\in \tm(\vv)\}$. From the construction  of it follows that $Y_{k-1} \subseteq X_{k-1}$.

    Next let $O_{k-1} = \bigcup_{\vv\in O_k} \{\Phi^{-1}_{q_{k-1}, t}(\Phi^{-1}_{q_k, 1-t}(\vv)) \ | \ t \in \tm(\vv)\}$. Clearly $O_{k-1}$is an open set of non-zero measure with $O_{k-1} \subseteq Y_{k-1}$. Moreover we can fix $\vv_{k-1} \in O_{k-1}$ such that $\vv_{k-1} = \Phi^{-1}_{q_k, 1-t}(\vv_k)$ for some $t \in \tm(\vv_k)$. Continuing this way we can find $\vv_j, O_j, Y_j$ for $1 \leq j \leq k$  (with $Y_k = X_k$) such that $ \traj = (q_1, \vv_1) (q_2, \vv_2) \ldots (q_k, \vv_k)$ is a trajectory and $\vv_j \in O_j \subseteq Y_j \subseteq \hcube_j$ for $1 \leq j \leq k$. From the construction of $\mc_{qt}$ it follows that $Y_j = [\traj](j)$ for $1 \leq j \leq k$. From Lemma \ref{lem:trajclose} it follows that $\pathmc$ and $\traj$ are compatible. It is also clear due to Lemma \ref{lem:trajclose} that $\traj$ is robust.

    \item Suppose $\traj = (q_1, \vv_1) (q_2, \vv_1) \ldots(q_k, \vv_k) \in \TRJ$ is robust. Then by Lemma \ref{lem:trajclose} there exist open sets $O_j$ of non zero measure and $\hcube_j \in \hyper$ such that $\vv_j \in O_j \subseteq [\traj](j) \subseteq \hcube_j$ for $1 \leq j \leq k$. Let $\traj^{(j)}$ denote the $j$-length prefix of $\traj$ for $1 \leq j \leq k$. We now define $X_j = [\traj^{(j)}](j)$ for $1 \leq j \leq k$. Then using the construction of $\mc_{qt}$ it is easy to show that there exists distributions $Pr_j$ over $X_j$ such that $\pathmc = \epsilon \eta_1 \eta_2 \ldots \eta_k$ is a path in $\mc_{qt}$ with  $\eta_j = (q_j, X_j, \hcube_j, Pr_j)$ for $1 \leq j \leq k$ and that $\pathmc$ is compatible with $\traj$.
\end{enumerate}
\end{proof}

We can now prove Theorem \ref{thm:qs}.

\begin{theorem}\label{thm:qs}\hspace{0.2cm} $H\models_R\psi$ iff $M_{qt}\models \psi$.
\begin{proof}
Suppose $H \not\models_R\psi$. Then there exists $\traj \in TRJ$ such
that $\traj$ is robust and $\traj, 0 \not\models_H \psi$. By Lemma \ref{lem:properties}, there exists a path $\pathmc$ in $M_{qt}$
which is compatible with $\traj$. Hence again by Lemma \ref{lem:properties} we then have $\pi \notin models_{M_{qt}}(\psi)$ which leads to
$Pr_{< 1}(\psi)$.
Next suppose that $Pr_{< 1}(\psi)$. Then there exists a path $\pathmc$ in $M_{qt}$ such that $\pathmc, 1 \not\models_{M_{qt}} \psi$.
By Lemma \ref{lem:properties}, there exists a robust trajectory $\traj$ which is compatible with $\pathmc$ and $\traj,0\not\models_H \psi$. This implies $H\not\models_R\psi$.
\end{proof}
\end{theorem}

Finally, we wish to show that the number of non-robust trajectories are negligible compared with the robust ones. Hence they do not contribute much towards the dynamics of $H$. For that we need the following lemma.

\begin{lemma}
Suppose $\traj = (q_0, \vv_0) (q_1, \vv_1) \ldots (q_k, \vv_k)$ is a non-robust
trajectory and $\traj^{(j)}$ is the $j$-length prefix of $\traj$ for $1 \leq j\leq k+1$. Let $\hcube_j = hc(\vv_j)$ and $Y_j = [\traj^{(j+1)}](j+1)$ for $0\leq j \leq k$. Then $Y_j$ is measurable and $Y_j \subseteq \hcube_j$ for $0 \leq j \leq k$. Furthermore $Y_j$ is of measure 0 for each $j$ in $\{0,1,..., k\}$.
\end{lemma}

\begin{proof}
Since $\traj$ is not robust, there exists $j: 0\leq j\leq k$ such that $\vv_j(i)=c_i\in C_i$ for some $i$ and hence for all $\vv\in \hcube_j$, $\vv(i)=c_i$ which implies $\mu(\hcube_j)=0$.  We induct on $j$. For $j=0$, $Y_0 = \INIT \cap \hcube_0$ is measurable and has measure 0. Suppose $q_0\move{g}q_1$ and let $Y'_1 = \bigcup_{\vv\in Y_0}\{\Phi_{q_1,1-t}(\Phi_{q_0,t}(\vv))\ |\ t\in \tm(\vv)\}$ where $\tm(\vv) = \{t\ |\ \Phi_{q_0,t}(\vv)\in g\}$. Then $Y_1 = Y'_1\cap \hcube_1$. Let $\hat Y_1 = \Phi_{q_1}((0,1)\times \Phi_{q_0}((0,1)\times Y_0)\cap g)$. Since $\mu(Y_0)=0$ hence $\mu((0,1)\times Y_0)=0$. Now both $\Phi_{q_1}$ and $\Phi_{q_0}$ are Lipschitz, and hence $\mu(\hat Y_1)=0$ [since the image of a set of measure 0 has measure 0 under a Lipschitz function]. Now note that $Y_1\subseteq \hat Y_1$ and hence $Y_1$ must be measurable and $\mu(Y_1)=0$. Continuing this way, we can show that $Y_j$ is measurable for all $j: 2\leq j\leq k$ and $\mu(Y_j)=0$.

Next suppose $j>0$. By a similar argument we can show that $Y_\ell$ is measurable for all $j<\ell\leq k$ and $\mu(Y_\ell)=0$. Let $q_{j-1}\move{g}q_j$ and let $Y'_{j-1} = \bigcup_{\vv\in Y_j}\{\Phi^{-1}_{q_{j-1},1-t}(\Phi^{-1}_{q_j,t}(\vv))\ |\ t\in \tm(\vv)\}$ where $\tm(\vv) = \{t\ |\ \Phi_{q_{j-1},t}(\vv)\in g\}$. Then $Y_{j-1}= Y'_{j-1}\cap \hcube_{j-1}$. Let $\hat Y_{j-1} = \Phi_{q_{j-1}}((-1,0)\times \Phi_{q_j}((-1,0)\times Y_j)\cap g)$. Since $\mu(Y_j)=0$ hence $\mu((-1,0)\times Y_j)=0$. Now both $\Phi_{q_j}$ and $\Phi_{q_{j-1}}$ are Lipschitz, and hence $\mu(\hat Y_{j-1})=0$ [since the image of a set of measure 0 has measure 0 under a Lipschitz function]. Now note that $Y_{j-1}\subseteq \hat Y_{j-1}$ and hence $Y_{j-1}$ must be measurable and $\mu(Y_{j-1})=0$. Continuing this way, we can show that $Y_m$ is measurable for all $m: 0\leq m <  j$ and $\mu(Y_m)=0$. \qed
\end{proof}

Thus by the above lemma, if a trajectory $\traj\in TRJ^{K+1}$ is not robust then there exists a $j\in \{0,1,\ldots,K\}$ such that $\mu(Y_j)=0$. This implies that in the product topology of $Q^{K+1} \times \reals^{K+1}$, $[\traj]$ has measure 0. Thus, the contribution made by the non-robust trajectories to the dynamics of $H$ is negligible.

Thus in terms of the sub-dynamics consisting of robust trajectories there is again a strong relationship between the behaviors of $H$ and $\mc_{qt}$. It also turns out that in measure-theoretic terms the non-robust trajectories can be ignored. More precisely if one starts with the discrete topology over $Q^{K+1}$ and the usual topology over $\Re^{n^{K+1}}$  one can easily define a natural measure space over the product topology $Q^{K+1} \times \Re^{n^{K+1}}$. In this space for every non-robust trajectory $\tau$ the representation of $[\tau]$ will be measurable but with measure $0$. In this sense the contributions made by the non-robust trajectories to the dynamics of $H$ are negligible.

\subsection*{Trajectory simulation for quantitative specifications}
Algorithm 3 gives the procedure for simulating robust trajectories for the verification of quatitative BLTL specifications. By Lemma \ref{lem:trajclose} , a trajectory is robust iff it does not hit any of the constants mentioned in the atomic propositions. The procedure is the same as Algorithm 1 before, except that whenever a value state $\vv_k$ at any time step $k$ hits a constant mentioned in any of the atomic propositions, we discard $\vv_k$ and start the simulation again from the value state of the previous time step.

\begin{algorithm}
\scriptsize
\renewcommand{\thealgorithm}{3}
\caption{Robust trajectory simulation}
\label{algtrajsimqt}
Input: Hybrid automaton $H=(Q, q_{in}, \{F_q(\bx)\}_{q\in Q}, \mathcal{G}, \ra, \INIT)$, maximum time step $K$.

Output: Trajectory $\tau$
	\begin{algorithmic}[1]
		\State Sample $\vv_0$ from $\INIT$ uniformly. If $\vv_0(i)\in C_i$ for any $i$, repeat.
                \State Set $q_0:=q_{in}$ and $\tau:=(q_0,\vv_0)$.
		\For{$k:=1 \ldots K$}
                    \Repeat
		        \State Generate time points  $T:=\{t_1,\ldots,t_J\}$ uniformly in $(0,1)$.
			\State Simulate $\vv^{\ell} := \Phi_{q_{k-1}}(t_{\ell},\vv_{k-1})$, for $\ell\in\{1,\ldots,J\}$
			\State Let $\widehat{\tm}_j:=\{t \in T: \vv^{\ell} \in g_j\}$ be the time points where $g_j$ is enabled.
			\State Pick $g_{\ell}$ randomly according to probabilities $p_j := \frac{|\widehat{\tm}_{j}|}{\sum_{j'=1}^m | \widehat{\tm}_{j'}|}$.
			\State Pick $t_{\ell}$ uniformly at random from $\widehat{\tm}_{\ell}$.
			\State Simulate $\vv':=\Phi_{q'}(1-t_{\ell},\vv^{\ell})$, where $q'$ is the target of $g_{\ell}$.
                     \Until{$\vv'(i)\notin C_i$ for any $i$}
			\State Set $q_k := q'$, $\vv_k := \vv'$, and extend $\tau := (q_0,\vv_0)\ldots(q_k,\vv_k)$.
		\EndFor
	\State \textbf{return} $\tau$
	\end{algorithmic}
\end{algorithm}

To see that the algorithm terminates with probability 1, note that if $\vv_0\in\hcube$ and $\hcube(i) = \{c\}$ for some $c \in C_i$ then $\mu(\hcube) = 0$. Thus Step 1 repeats with probability 0. As a result with probability $1$ it will be repeated only a finite number of times. Similarly the repeat loop of Step 4-11 will terminate with probability 1.

\section{Case studies}\label{sec:results}
We first evaluated our method on a model of the electrical dynamics of the cardiac cell \cite{fenton08}. We also applied our method on a model of circadian rhythm network \cite{miyano06}. The $\Delta$ time step parameter for the cardiac cell model and the circadian rhythm model were both set to $0.1$. The parameters used for the statistical model checking were $\delta=0.01$ and $\alpha=0.01$. We have implemented our method using MATLAB. The source code is available at \url{http://github.com/bgyori/hybrid}. The experiments were carried out on a PC with a 3.4GHz Intel Core i7 processor with 8GB RAM. Simulating one trajectory took, on average, $5.2$s for the circadian clock model and $18.3$s for the cardiac cell model. We note that when checking quantitative properties, the trajectories that hit corner points such as $u = 1.4$ will be non-robust and hence can be ignored. Our implementation exploits the parallelization enabled by statistical model checking, hence multiple trajectories can be simulated simultaneously. A summary of the results for the verification of all properties for both models, along with the number of samples taken to complete the verification is given in Table 3 of the Appendix. 

In our experiments, we used $J=10$ as the number of intermediate time steps for choosing mode transitions. We investigated whether this choice is sufficient for accurate simulation. We simulated $1000$ independent realizations of the cardiac cell system with $J=10$ and $J=100$, and compared the distributions of the modes that the system is in at a series of discrete time points. The Kolmogorov-Smirnov statistical test did not reject the hypothesis that the two distributions are the same (at confidence level $95\%$). This indicates that using $J=10$ is adequate.

\subsection{Cardiac cell model}
Heart rhythm depends on the organized opening and closing of gates--called ion channels--on the cell membrane, which govern the electrical activity of cardiac cells.
Disordered electric wave propagation in heart muscle can cause cardiac abnormalities such as \textit{tachycardia} and \textit{fibrillation}. The dynamics of the electrical activity of a single human ventricular cell has been modeled as a hybrid automaton \cite{fenton08,grosu11} shown in Figure \ref{heart}. The model contains $4$ state variables and $26$ parameters. Ventricular cells consist of three subtypes, namely epicardial, endocardial, and midmyocardial cells, which possess different dynamical characteristics. The cell-type-specific parameters of the model are summarized in Table \ref{parameter} in the Appendix.
An action potential (AP) is a change in the cell's transmembrane potential $u$, as a response to an external stimulus (current) $\epsilon$.
The flow of total currents is controlled by a fast channel gate $v$ and two slow gates $w$ and $s$.

\begin{figure}[ht]
\centering
\includegraphics[width=0.9\textwidth]{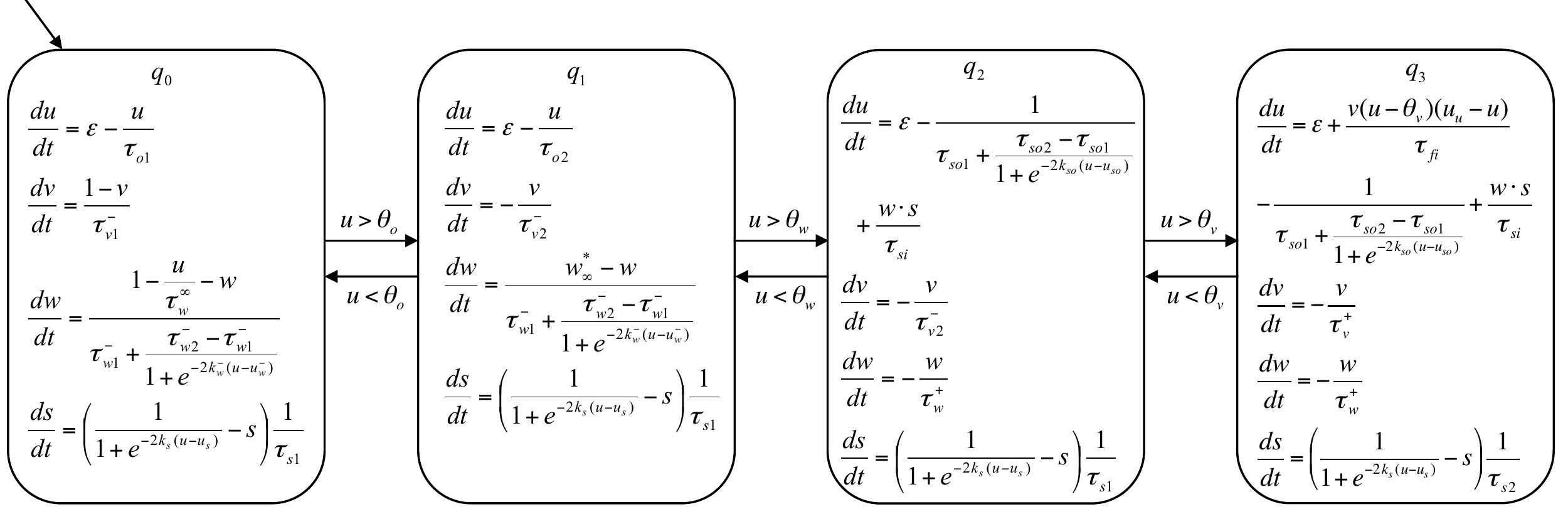}
\caption{The hybrid automaton model for the cardiac cell system \cite{grosu11}.}
\label{heart}
\end{figure}

In mode $q_0$, the ``Resting mode'', the cell is waiting for stimulation. We assume an external stimulus $\epsilon$ equal to $1$ mV lasting for $1$ millisecond. The stimulation causes $u$ to increase which may trigger a mode transition to mode $q_1$. In mode $q_1$, gate $v$ starts closing and the decay rate of $u$ changes. The system will jump to mode $q_2$ if $u > \theta_w$. In mode $q_2$, gate $w$ is also closing.
When $u > \theta_v$, mode $q_3$ can be reached, which means a successful ``AP initiation''. In mode $q_3$, $u$ reaches its peak due to the fast opening of a sodium channel. The cardiac muscle then contracts and $u$ starts decreasing.

\noindent\textbf{Property C1}
It is known that the cardiac cell can lose its excitability, which will lead to disorders such as ventricular tachycardia and fibrillation. We formulated the property for responding to stimulus by leaving the resting mode:
\begin{center}
$\mathbf{F}^{\leq 500}(\neg[\text{Resting mode}])$.
\end{center}
The property was verified to be \emph{true} for all three cell types  under the healthy condition. However, under a disease condition (for example $\tau_{o1} = 0.004$ or  $\tau_{o2} = 0.1$ \cite{delta}) the property was verified to be \emph{false} no matter what stimulation value of $\epsilon$ was used. Consequently, a region of such unexcitable cells blocks the impulse conduction and can lead to cardiac disorders such as fibrillation. This is consistent with experimental results reported in \cite{tanaka07}.

\noindent\textbf{Property C2} After successfully generating an AP (that is, reaching the ``AP mode'', $q_3$), the cardiac cell should return to a low transmembrane potential and wait in ``Resting mode'' for the next stimulation. The corresponding formula is
\begin{center}
$\mathbf{F}^{\leq 500}([\text{AP mode}]) \wedge \mathbf{F}^{\leq 500}(\mathbf{G}^{\leq 100}([\text{Resting mode}]))$.
\end{center}
The above query was verified to be \emph{true} for all three cell types  under the healthy condition and transient stimulation. However, if we change the stimulation profile from transient to sustained, i.e. assuming $\epsilon$ lasts for $500$ milliseconds, the property was verified to be \emph{false}--the cell doesn't return to and settle at a low transmembrane potential resting state. In ventricular tissue the stimulus $\epsilon$ can be delivered from neighboring cells \cite{fenton08}. Thus, our results suggest that the transient activation of a single cardiac cell depends on the stimulation profile of its neighboring cells. 

\noindent\textbf{Property C3} It has been reported that epicardial, endocardial, and midmyocardial cells have different AP morphologies \cite{nabauer96,drouin95}. In particular, a crucial ``spike-and-dome'' AP morphology can only be observed in epicardial cells but not endocardial and midmyocardial cells (see Figure \ref{morphology} of the Appendix). 
We formulated the property for a spike-and-dome AP morphology as a quantitative property,
\begin{center}
$\mathbf{F}^{\leq 500}(\mathbf{G}^{\leq 1}([1.4 \leq u]) \wedge \mathbf{F}^{\leq 500}([0.8 \leq u] \wedge [u \leq 1.1] \wedge \mathbf{F}^{\leq 500}(\mathbf{G}^{\leq 50}([1.1 \leq u]))))$.
\end{center}
The property was verified to be \emph{true} for epicardial cell, and  \emph{false} for endocardial and midmyocardial cells, under the healthy condition and transient stimulation. Among $26$ model parameters, $20$ of them have different values over different cell types. We then perturbed each epicardial parameter and checked if the above property still holds. Our results show that $\tau_{s2}$ is key to the AP morphology (i.e. the spike-and-dome AP morphology disappears when $\tau_{s2}=2$), which highlights the importance of $s$ gate to epicardial cells. This is consistent with \cite{delta} that the model proposed in \cite{fenton98}, which does not includes $s$ gate, is unable to capture the dynamics of epicardial cells.

\subsection{Circadian rhythm model}
Mammalian cells follow a circadian rhythm with a 24h period, which is generated and governed by a highly coupled transcription-translation network. The model diagram and the corresponding hybrid system dynamics proposed in \cite{miyano06,miyano09} is shown in the Appendix. The system comprises 16 modes, each of which contains 12 state variables and 29 parameters. Each mode corresponds to a particular combination of ON or OFF transcriptional states of genes \emph{Per}, \emph{Cry}, \emph{Rev-Erb}, \emph{Clock}, and \emph{Bmal}. The switches between modes are guarded by the threshold levels of protein complexes PER-CRY, CLOCK-BMAL and REV-REB. The mRNA levels of \emph{Per} and \emph{Cry} are known to be oscillating due to the negative feedback loops in the network. Specifically, there are two major negative feedback (NF) loops: (i) the core NF formed by PER-CRY, CLOCK-BMAL, PER, and CRY and (ii) a complement NF formed by REV-ERB, BMAL, and CLOCK-BMAL. The time constants appearing in the properties are in minute units. 

\noindent\textbf{Property R1} Similar to \emph{Per} and \emph{Cry}, the expression level of \emph{Bmal} gene is also oscillating \cite{shearman00}. We formulated this property as
\begin{center}
$\mathbf{F}^{\leq 500}([1.5 \leq \text{\emph{Bmal}}] \wedge \mathbf{F}^{\leq 500}([\text{\emph{Bmal}} \leq 0.8] \wedge \mathbf{F}^{\leq 500}([1.5 \leq \text{\emph{Bmal}}] \wedge \mathbf{F}^{\leq 500}([\text{\emph{Bmal}} \leq 0.8] \wedge \mathbf{F}^{\leq 500}([1.5 \leq \text{\emph{Bmal}}])))))$
\end{center}
The property was verified to be \emph{true} under the wild type condition. It was verified to be \emph{false} under \emph{Cry} mutant condition but \emph{true} in the \emph{Rev-Erb} mutant condition, which is consistent with the experimental data in \cite{kim12,shearman00}. This suggests that the oscillatory behavior of \emph{Bmal} mRNA is induced by the core negative feedback mediated by PER-CRY, instead of the complement negative feedback mediated by REV-ERB.

\noindent\textbf{Property R2} It has been observed that the peaks of \emph{Bmal} mRNA are always located between two successive \emph{Per} or \emph{Cry} mRNA peaks \cite{kim12}. The corresponding formula is
\begin{center}
$\mathbf{F}^{\leq 500}([\text{\emph{Bmal}} \leq 0.8] \wedge [2.0 \leq \text{\emph{Per}}] \wedge [2.0 \leq \text{\emph{Cry}}] \wedge  \mathbf{F}^{\leq 500}([1.5 \leq \text{\emph{Bmal}}] \wedge [\text{\emph{Per}} \leq 0.8] \wedge [\text{\emph{Cry}} \leq 0.8] \wedge  \mathbf{F}^{\leq 500}([\text{\emph{Bmal}} \leq 0.8] \wedge [2.0 \leq \text{\emph{Per}}] \wedge [2.0 \leq \text{\emph{Cry}}] \wedge  \mathbf{F}^{\leq 500}([1.5 \leq \text{\emph{Bmal}}] \wedge [\text{\emph{Per}} \leq 0.8] \wedge [\text{\emph{Cry}} \leq 0.8]))))$
\end{center}
The above query was verified to be \emph{true} under wild type condition. If we remove the dependence between \emph{Bmal} transcription and PER-CRY concentration, the property R2 was verified to be \emph{false}, while the property R1 was verified to \emph{true} (i.e. oscillating). Thus, our results suggest that the complement negative feedback mediated by REV-ERB is responsible for maintaining the oscillatory behavior of \emph{Bmal} mRNA level while PER-CRY plays a role in delaying the \emph{Bmal} expression responses.
\vspace{-0.2cm}
\section{Conclusion}
We have presented an approximate  probabilistic verification method for analyzing the dynamics of a hybrid system $H$ in terms of a Markov chain $M$. For bounded time properties, we have shown a strong correspondence between the behaviors of $H$ and $M$. We have also extended this result to handle quantitative atomic propositions and shown a similar correspondence result for the sub-dynamics consisting of robust trajectories.
Thus the intractable verification problem for $H$ can be solved approximately using its Markov chain approximation. Accordingly, we have devised  a statistical model checking procedure to verify that $M$ almost certainly meets a BLTL specification and then applied this procedure to two examples to demonstrate the applicability of our approximation scheme. A hardware accelerated parallel implementation of the trajectory sampling procedure will considerably improve the performance and scalability of our method. Overall, we view our results as providing a mathematical basis for verifying if a hybrid system models satisfies a BLTL property with high probability. 

As an extension, one could consider more sophisticated  stochastic assumptions regarding the time points and value states at which the mode transitions take place. These assumptions will however have to be justified and motivated by the modeling problem at hand, especially in systems biology applications. 
Yet another valuable extension will be to study  a network of hybrid systems. This will enable us to model the cross talk, feed-forward and feed-back loops involving multiple signaling pathways.

\bibliographystyle{abbrv}
\bibliography{sigproc}

\clearpage
\section*{Appendix}

\subsection*{Case studies}
The equations governing the dynamics of the circadian clock model are given in Figure \ref{circ_model}. The equations contain rate constants, which are denoted $k_1$ to $k_28$, set according to \cite{miyano09}. The combination of ``mode indicator'' binary variables $\theta_{CB}$ to $\theta_{RE}$, $\theta_{PC1}$, $\theta_{PC2}$ and $\theta_{PC3}$ define the mode of the dynamics, and each mode is defined by a unique value combination of the mode indicators. These value combinations are listed in Table \ref{circ_modes}. The guards associated with a source and target mode are constructed as follows. Each mode indicator corresponds to a guard component, which is a threshold on a state variable. For instance, $\theta_{RE}$ has the corresponding guard component [REV-ERB]$<1.1$. The guard to a target mode is enabled if all the mode indicators that are on in the mode are enabled according to their respective guard components. Finally, a transition between a source and a target mode only exists if there is only one difference in the combination fo mode indicators. For instance, there is a transition from mode 1 to mode 2 but not from mode 1 to mode 9. The dynamics of the \emph{Clock} mRNA is governed externally.

\begin{figure}
	\centering
	\includegraphics[width=0.9\textwidth]{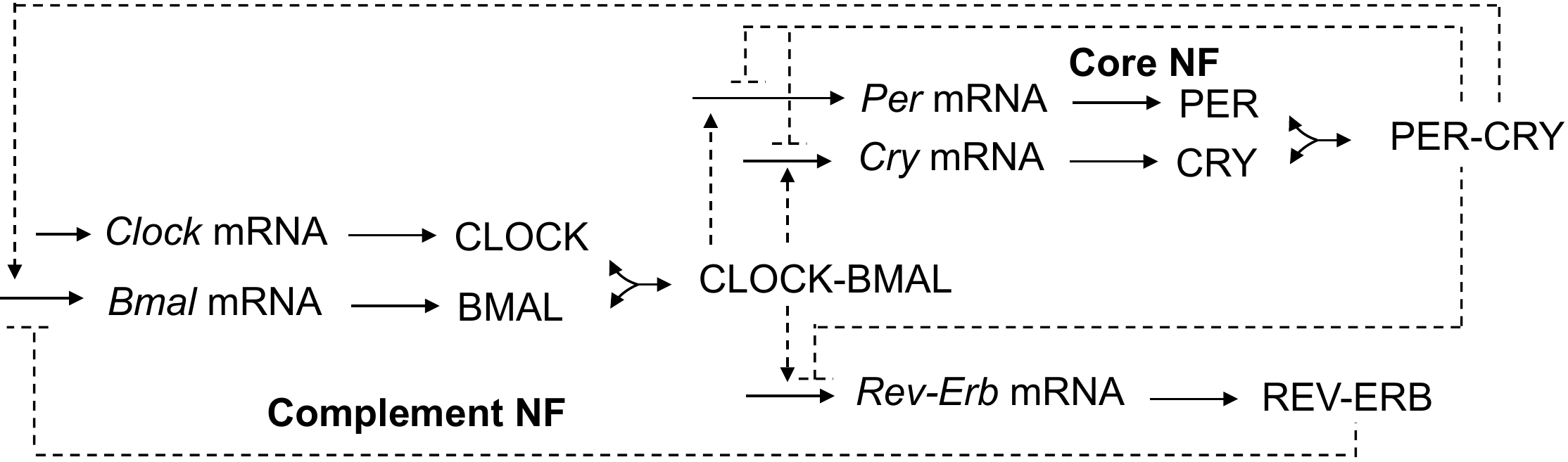}
	
	\vspace{10pt}
	
	\includegraphics[width=0.4\textwidth]{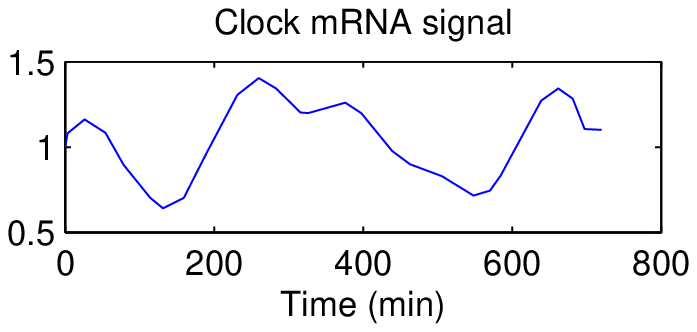}
	
	\vspace{-10pt}
	
	\begin{align*}\label{circ_ode}
		\text{d}/\text{d}t[\text{Per}] &= - k_{1}\cdot [\text{Per}] + k_{13}\cdot \theta_{PC2}\cdot \theta_{CB} + k_{14}\\
		\text{d}/\text{d}t[\text{PER}] &= - k_{2}\cdot [\text{PER}] + k_{15}\cdot [\text{Per}] - k_{16}\cdot [\text{PER}]\cdot [\text{CRY}]\\
		\text{d}/\text{d}t[\text{Cry}] &= - k_{3}\cdot [\text{Cry}] + k_{17}\cdot \theta_{PC2}\cdot \theta_{CB} + k_{18}\\
		\text{d}/\text{d}t[\text{CRY}] &= - k_{4}\cdot [\text{CRY}] + k_{19}\cdot [\text{Cry}] - k_{16}\cdot [\text{PER}]\cdot [\text{CRY}]\\
		\text{d}/\text{d}t[\text{PER-CRY}] &= - k_{5}\cdot [\text{PER-CRY}] + k_{16}\cdot [\text{PER}]\cdot [\text{CRY}]\\
		\text{d}/\text{d}t[\text{Rev-Erb}] &= - k_{6}\cdot [\text{Rev-Erb}] + k_{20}\cdot \theta_{PC1}\cdot \theta_{CB} + k_{21}\\
		\text{d}/\text{d}t[\text{REV-ERB}] &= - k_{7}\cdot [\text{REV-ERB}] + k_{22}\cdot [\text{Rev-Erb}]\\
		\text{d}/\text{d}t[\text{CLOCK}] &= - k_{9}\cdot [\text{CLOCK}] + k_{24}\cdot [\text{Clock}] - k_{25}\cdot [\text{CLOCK}]\cdot [\text{BMAL}]\\
		\text{d}/\text{d}t[\text{Bmal}] &= - k_{10}\cdot [\text{Bmal}] +  k_{26}\cdot \theta_{PC3}\cdot \theta_{RE} + k_{27}\\
		\text{d}/\text{d}t[\text{BMAL}] &= - k_{11}\cdot [\text{BMAL}] + k_{28}\cdot [\text{Bmal}] - k_{25}\cdot [\text{CLOCK}]\cdot [\text{BMAL}]\\
		\text{d}/\text{d}t[\text{CLOCK-BMAL}] &= - k_{12}\cdot [\text{CLOCK-BMAL}] + k_{25}\cdot [\text{CLOCK}]\cdot [\text{BMAL}]
	\end{align*}
	\caption{The model diagram, the Clock mRNA signal and the equations governing the circadian clock model.}
	\label{circ_model}
\end{figure}

\begin{table}
	\centering
	
	\begin{tabular}{c|c}
		Mode indicator & Guard component \\
		\hline
		$\theta_{RE}$ & [REV-ERB]$<1.1$ \\
		$\theta_{CB}$ & [CLOCK-BMAL]$>1.0$ \\
		$\theta_{PC1}$ & [PER-CRY]$<1.4$ \\
		$\theta_{PC2}$ & $1.4<$[PER-CRY]$<1.5$ \\
		$\theta_{PC3}$ & $2.2<$[PER-CRY] \\
	\end{tabular}
	
	\vspace{10pt}
	
	\begin{tabular}{c|cccc}
		Mode & 1 & 2 & 3 & 4  \\
		$(\theta_{PC1},\theta_{PC2},\theta_{PC3},\theta_{RE},\theta_{CB})$ &
		(1,1,0,1,0) & (1,1,0,1,1) & (1,1,0,0,0) & (1,1,0,0,1) \\
		\hline
		Mode & 5 & 6 & 7 & 8  \\
		$(\theta_{PC1},\theta_{PC2},\theta_{PC3},\theta_{RE},\theta_{CB})$ &
		(0,1,0,1,0) & (0,1,0,1,1) & (0,1,0,0,0) & (0,1,0,0,1) \\
		\hline
		Mode & 9 & 10 & 11 & 12 \\
		$(\theta_{PC1},\theta_{PC2},\theta_{PC3},\theta_{RE},\theta_{CB})$ &
		(0,0,0,1,0) & (0,0,0,1,1) & (0,0,0,0,0) & (0,0,0,0,1) \\
		\hline
		Mode & 13 & 14 & 15 & 16  \\
		$(\theta_{PC1},\theta_{PC2},\theta_{PC3},\theta_{RE},\theta_{CB})$ &
		(0,0,1,1,0) & (0,0,1,1,1) & (0,0,1,0,0) & (0,0,1,0,1)
	\end{tabular}
	
	\vspace{10pt}
	
	\caption{The $5$ mode indicator variables and their associated guard components (top). The 16 modes of the circadian clock model with the corresponding combination of binary mode indicator variables (bottom).}
	\label{circ_modes}
\end{table}

The parameters used for the cardiac cell model are given in Table \ref{parameter}.

\begin{table}[htb]
	\centering
	\small
	\begin{tabular}{|c|ccc|c|ccc|}
		\hline
		Parameter & EPI & ENDO & MID  & Parameter & EPI & ENDO & MID  \\ \hline
		$\theta_o$ & $0.006$ & $0.006$ & $0.006$ & $\tau_{v1}^-$ & $60$ & $75$ & $80$\\
		$\theta_w$ & $0.13$ & $0.13$& $0.13$ & $\tau_{v2}^-$ & $1150$ & $10$& $1.4506$ \\
		$\theta_v$ & $0.3$ & $0.3$& $0.3$ & $\tau_{w1}^-$ & $60$ & $6$& $70$ \\
		$u^-_{w}$ & $0.03$ & $0.016$& $0.016$ & $\tau_{w2}^-$ & $15$ & $140$& $8$ \\
		$u_{so}$ & $0.65$ & $0.65$& $0.6$ & $\tau_{o1}$ & $400$ & $470$& $410$\\
		$u_s$ & $0.9087$  & $0.9087$& $0.9087$ & $\tau_{o2}$ & $6$ & $6$& $7$ \\
		$u_u$ & $1.55$ & $1.56$& $1.61$ & $\tau_{so1}$ & $30.0181$ & $40$& $91$\\
		$w_{\infty}^{*}$ & $0.94$ & $0.78$& $0.5$ & $\tau_{so2}$ & $0.9957$ & $1.2$& $0.8$\\
		$k_w^-$ & $65$ & $200$& $200$ & $\tau_{s1}$ & $2.7342$ & $2.7342$& $2.7342$\\
		$k_{so}$ & $2.0458$ & $2$& $2.1$ & $\tau_{s2}$ & $16$ & $2$& $4$\\
		$k_s$ & $2.994$ & $2.994$& $2.994$ & $\tau_{fi}$ & $0.11$ & $0.1$& $0.078$\\
		$\tau_v^+$ & $1.4506$ & $1.4506$& $1.4506$ & $\tau_{si}$ & $1.8875$ & $2.9013$& $3.3849$\\
		$\tau_w^+$ & $200$ & $280$& $280$ & $\tau_{w\infty}$ & $0.07$ & $0.0273$& $0.01$\\
		\hline\end{tabular}
	\vspace{10pt}
	\caption{Parameter values of the cardiac model for epicardial (EPI), endocardial (ENDO), and midmyocardial (MID) cells under healthy condition.}
	\label{parameter}
\end{table}

\begin{figure}[ht]
	\centering
	\includegraphics[width=1.0\textwidth]{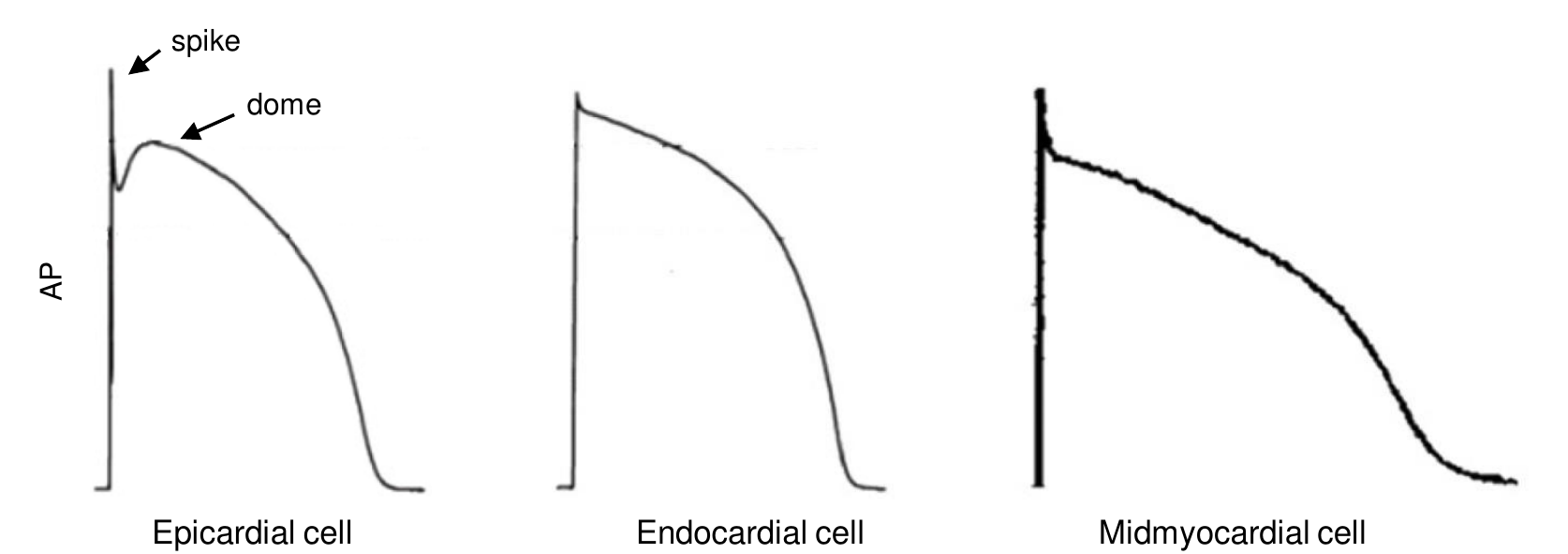}
	\caption{The AP morphologies of epicardial \cite{nabauer96}, endocardial \cite{nabauer96} and midmyocardial \cite{drouin95} cells.}
	\label{morphology}
\end{figure}

\begin{table}[htb]
	\centering
	\scriptsize
	\begin{tabular}{|c|c|c|p{1.5cm}|} \hline
		Property & Condition & Decision & \# samples before stopping \\ \hline
		C1 & Epicardial, Healthy & True & 459 \\
		C1 & Endocardial, Healthy & True & 459  \\
		C1 & Midmyocardial, Healthy & True & 459  \\
		C1 & Epicardial, Diseased & False & 1  \\
		C1 & Endocardial, Diseased & False & 1  \\
		C1 & Midmyocardial, Diseased & False & 1  \\
		C2 & Epicardial, Transient & True & 459  \\
		C2 & Endocardial, Transient & True & 459  \\
		C2 & Midmyocardial, Transient & True & 459  \\
		C2 & Epicardial, Sustained & False & 1  \\
		C2 & Endocardial, Sustained & False & 1  \\
		C2 & Midmyocardial, Sustained & False & 1  \\
		C3 & Epicardial, $\tau_{s2}=16$ & True & 459  \\
		C3 & Epicardial, $\tau_{s2}=2$ & False & 1  \\
		C3 & Endocardial & False & 1  \\
		C3 & Midmyocardial & False & 1  \\
		R1 & Wild type & True & 459  \\
		R1 & Cry mutant & False & 1  \\
		R1 & Rev-Erb mutant & True & 459  \\
		R2 & Wild type & True & 459   \\
		R2 & Without PER-CRY dependence  & False & 1  \\
		R1 & Without PER-CRY dependence  & True & 459  \\
		\hline\end{tabular}
	\vspace{10pt}
	\caption{Results summary of SMC for hybrid systems}\label{performance}.
\end{table}

\end{document}